\documentclass[11pt]{article}
\usepackage[a4paper, margin=1in]{geometry}
\usepackage[utf8]{inputenc}
\usepackage{lipsum, xcolor}
\usepackage{amsmath,amssymb,amsthm}
\usepackage{graphicx}
\usepackage{natbib}

\newtheorem{theorem}{Theorem}[section]
\newtheorem{lemma}[theorem]{Lemma}
\newtheorem{corollary}[theorem]{Corollary}
\newtheorem{assumption}{Assumption}
\newtheorem{proposition}{Proposition}
\newtheorem{example}{Example}
\newtheorem{remark}{Remark}
\usepackage{enumitem}
\usepackage{titlesec}
\usepackage{booktabs}

\newcommand{\N}{\mathcal{N}}
\newcommand{\R}{\mathbb{R}}
\newcommand{\PE}{\mathrm{PE}}
\newcommand{\cv}{\mathrm{CV}}
\newcommand{\CB}{\mathrm{CB}}

\newcommand{\tY}{{\widetilde{Y}}}
\newcommand{\calL}{{\mathcal{L}}}
\newcommand{\bbP}{{\mathbb{P}}}
\newcommand{\EE}[2][]{\mathbb{E}_{#1}\left[#2\right]}
\newcommand{\PP}[2][]{\mathbb{P}_{#1}\left[#2\right]}
\newcommand{\Cov}[2][]{\operatorname{Cov}_{#1}\left[#2\right]}
\newcommand{\Var}[2][]{\operatorname{Var}_{#1}\left[#2\right]}
\newcommand{\iid}{\stackrel{i.i.d.}{\sim}}
\newcommand{\tran}{^\top}
\newcommand{\om}{\omega}
\newcommand{\numberthis}{\addtocounter{equation}{1}\tag{\theequation}}
\newcommand{\rd}{\mathrm{d}}
\newcommand\indep{\protect\mathpalette{\protect\independenT}{\perp}}
\def\independenT#1#2{\mathrel{\rlap{$#1#2$}\mkern2mu{#1#2}}}
\newcommand{\tr}{\operatorname{tr}}
\newcommand{\Rom}[1]{\uppercase\expandafter{\romannumeral #1\relax}}
\newcommand{\cov}{\mathrm{Cov}}
\newcommand{\frakA}{{\mathfrak{A}}}
\newcommand{\frakg}{{\mathfrak{g}}}

\newcommand{\calT}{{\mathcal{T}}}

\def\argmin{\mathop{\rm argmin}\limits}
\newcommand{\ep}{\varepsilon}

\renewcommand{\tilde}[1]{\widetilde{#1}}

\newcommand{\Indc}[1]{{\mathbf{1}\left\{{#1}\right\}}}
\usepackage{xr-hyper}
\usepackage[hidelinks]{hyperref}
\usepackage{subcaption}

\usepackage{authblk}

\title{{Cross-Validation with Antithetic Gaussian Randomization}}
\author[1]{Sifan Liu}
\author[2]{Snigdha Panigrahi} 
\author[2]{Jake A. Soloff}
\affil[1]{Department of Statistical Science, Duke University}
\affil[2]{Department of Statistics, University of Michigan}

\begin{document}

\maketitle

\begin{abstract}
We introduce a new cross-validation method based on an equicorrelated Gaussian randomization scheme. Our method is well-suited for problems where sample splitting is infeasible, either because the data violate the assumption of independent and identically distributed samples, or because there are insufficient samples to form representative train--test data pairs. In such problems, our method provides a simple, principled, and computationally efficient approach to estimating prediction error, often outperforming standard cross-validation while requiring only a small number of repetitions.

Drawing inspiration from recent splitting techniques like data fission and data thinning, our method constructs train--test data pairs using Gaussian randomization. Our main contribution is the introduction of an \emph{antithetic Gaussian randomization} scheme, involving a carefully designed correlation structure among the randomization variables. We show theoretically that this antithetic construction can eliminate the bias of cross-validation for a broad class of smooth prediction functions, without inflating variance. Through simulations across a range of data types and loss functions, we demonstrate that our estimator outperforms existing methods for prediction error estimation.
\end{abstract}

\section{Introduction}
\label{sec:1}

Estimating prediction error is a fundamental task in statistics and machine learning, employed to evaluate generalization ability, select tuning parameters, and compare competing models. Cross-validation is one of the most widely used tools for this purpose. In its most standard form, cross-validation partitions the data into independent subsamples or ``folds,'' and prediction error is estimated by averaging the empirical errors from the test folds. The popularity of cross-validation is easy to understand: it is assumption-light and versatile, accommodating a wide range of loss functions and data types.

Standard cross-validation, however, is not suitable for all types of data, especially when i.i.d. assumptions are violated or when the sample size is too small to form representative folds.
Many examples of such data arise across different domains.
In time series or spatially correlated data, splitting the sample space into train--test folds can disrupt inherent dependence structure.
In regression settings with influential observations, folds that exclude these points may fail to adequately represent the full data.
With categorical variables, sample splitting may yield imbalanced folds or omit rare categories entirely. 
In such cases, standard cross-validated estimates of prediction error can be severely misleading and can result in unreliable models for downstream tasks.

In this paper, we introduce a novel cross-validation method for estimating prediction error without splitting the sample space.
Instead, train--test folds in our method are created using equicorrelated Gaussian randomization variables that sum to zero.
Although the form of randomization in our method is quite different from the subsampling scheme employed in standard cross-validation, it retains a key property of cross-validation: the original data can be reconstructed by pooling the folds. 
Next, we present the basic framework, highlight our contributions, and review related work.

\section{Framework, contributions, and related work}
\label{sec:2}
\subsection{Overview of the basic framework and instantiating our method}

We assume the response vector $Y=(Y_1,\ldots,Y_n)\tran\in\R^n$ is drawn from a distribution $\bbP_n$, while any predictors, when present, are treated as fixed.
Given a real-valued loss function $\calL(\theta, Y)$, where $\theta$ is an unknown parameter, and a fitted estimator or prediction function $g$, our goal is to evaluate its performance on unseen test data $\tY$, an independent copy of the observed data $Y$.
Specifically, our estimand of interest is the expected prediction error, defined as
\begin{equation}
\PE(g)=\EE{\calL(g(Y), \tY )},
\label{estimand: pe}
\end{equation}
also called generalization error or test error.
Here, the expectation is taken over both the train data $Y$ and the test data $\tY$. The prediction function $g$ may depend on predictors, but we suppress this dependence in the notation since predictors are treated as fixed.

Our framework assumes that the prediction function $g$
depends on the data $Y$ only through a \emph{sufficient} statistic $S(Y)$.
Importantly, this framework does not rely on i.i.d. assumptions; rather, it assumes that the sufficient statistic $S(Y)$ is approximately normal.
As illustrated by examples throughout the paper, our framework accommodates many commonly used loss functions, including those applied in fitting generalized linear models (GLMs).

As a starting example, consider normal data $Y \sim \N(\theta,\sigma^2I_n) \in \R^n$, with a known variance $\sigma^2$. 
For quadratic loss function $\calL(\theta, Y)=\|\theta-Y\|^2_2$, our goal is to estimate the prediction error
\begin{equation}
\PE(g)= \EE{\|g(Y)- \tY\|_2^2}.
\label{pred:error}
\end{equation}
Our cross-validation method for this problem proceeds as follows: we generate random vectors $\omega^{(k)}\sim \N(0,\sigma^2 I_n)$, for $k\in [K]=\{1,2,\ldots, K\}$, where the $K$ randomization variables are equicorrelated with each other.
As we formalize later, the correlation between any pair of these randomization variables is set to its most negative value possible.
Following the Monte Carlo literature on variance reduction techniques, e.g., \cite{craiu2005multiprocess}, we view our randomization scheme as an ``extreme antithesis,'' where  correlations between the added randomization variables take the most negative value possible. 

By adding a $\sqrt\alpha$-scaled version of the randomization variables, $\{ \omega^{(k)}: k\in [K]\}$, to the sufficient statistic—here, $S(Y)=Y$—we create $K$ train--test data pairs. 
Specifically, the $k$-th train--test data pair in our cross-validation method are constructed as:
\begin{align}\label{eq:simple-split}
 \left(Y_{\text{train}}^{(k)} =Y + \sqrt\alpha\omega^{(k)},\quad Y_{\text{test}}^{(k)}= Y - \frac{1}{\sqrt\alpha}\omega^{(k)}\right), \text{ for } k\in [K].
\end{align}
Then as done in standard cross-validation, for the $k$-th train--test repetition, we use $Y_{\text{train}}^{(k)}$ to fit the prediction function and $Y_{\text{test}}^{(k)}$ to estimate its performance. 

\subsection{Highlights of our method}

The performance of any cross-validation method, measured by mean squared error (MSE) of the corresponding estimator of the prediction error, is governed by a bias--variance tradeoff. For example, in standard cross-validation, this bias--variance tradeoff is controlled by the number of folds.

As illustrated in Equation~\eqref{eq:simple-split}, our cross-validation method depends on two user-specified parameters: a positive scalar $\alpha \in  \mathbb{R}^+$ and an integer $K \geq 2$. 
The first parameter, $\alpha$, is akin to the proportion of held-out samples in standard cross-validation. 
The second parameter, $K$, specifies the number of train--test repetitions over which estimates of prediction error are averaged.
The two parameters act as distinct levers in our method to control the bias and variance of the resulting prediction error estimator, making our cross-validation method particularly appealing for the following reasons.
\begin{enumerate}
\item \textbf{Direct control of bias via $\boldsymbol{\alpha}$:} The parameter $\alpha$ controls the bias introduced by estimating the prediction function on noisier training data, with the bias decaying to $0$ as $\alpha$ decreases. Unlike standard cross-validation, where bias is controlled by the number of folds, the parameter $\alpha$ in our method does not depend on the number of train--test repetitions, $K$. 

Separating the roles of these two parameters provides a significant advantage: by averaging empirical estimates of prediction error over just $K$ train--test repetitions---where $K$ can be as few as two---our method, with a small $\alpha$, can achieve very low bias. 
In contrast, achieving a similarly low bias with standard cross-validation requires the computationally intensive leave-one-out (LOO) method with $n$ train--test repetitions.

\item \textbf{Stable variance for finite $\mathbf{K}$:} A key strength of our estimator, as supported by our theoretical and empirical analysis, is its stable variance for a large class of smooth prediction functions, even as we let $\alpha$ vanish to~$0$. 

As a result, our estimator often achieves lower MSE compared to standard cross-validation, where reducing bias can come at the cost of increased variance. The stable variance of our method is due to the carefully designed ``antithetic'' correlation structure among the external Gaussian randomization variables.
\end{enumerate}

To the best of our knowledge, this work is the first to investigate the potential of an antithetic Gaussian randomization approach for cross-validation.
Figure~\ref{fig: isotonic mse} compares the MSE of our cross-validated estimator with standard cross-validation estimators for computing the prediction error \eqref{estimand: pe} in an isotonic regression problem.
Given that sample splitting is a widely used strategy for creating independent subsets, we adopt a similar approach to \cite{chaudhuri2023cross} here to implement a variant of the standard cross-validation procedure.
Our method uses only two train--test repetitions ($K=2$) with $\alpha=0.01$, while standard cross-validation is performed with $K=2$ folds and $K=100$ folds, the latter corresponding to leave-one-out (LOO) cross-validation.  
Remarkably, our estimator achieves a smaller MSE than LOO cross-validation while being $50$ times more computationally efficient. Additional details on this example and extensive numerical results examining the effects of $\alpha$ and $K$ are presented later in Section~\ref{sec: experiments}.

\begin{figure}
  \centering
  \includegraphics[width=0.4\textwidth]{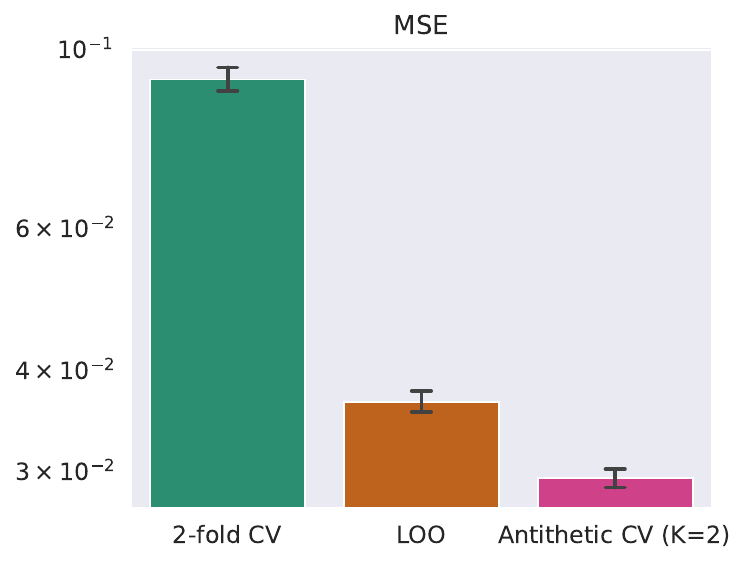}
  \caption{Estimating prediction error in an isotonic regression problem. The $y$-axis reports mean squared error (MSE), with error bars indicating variability across 1,000 independent simulations using synthetic data. From left to right, the methods shown are standard 2-fold CV, LOO CV, and the proposed antithetic CV method with $K=2$ and $\alpha=0.01$. Additional details are provided in Section~\ref{sec: experiments}.}
  \label{fig: isotonic mse}
\end{figure}

\subsection{Related work and contributions}

The idea of employing alternative forms of randomization in cross-validation is by no means new. 
\cite{brown2013poisson} introduced a ``nonstandard cross-validation method'' for the Gaussian sequence model, in which a single train--test split of the form~\eqref{eq:simple-split} is used for both estimation and hyperparameter tuning. 
This construction is related to our proposal in the special case of two folds.

Among more recent work, our proposal is closely related to the coupled bootstrap (CB) estimator introduced by \cite{oliveira2021unbiased} for normal data, which we review in the next section. 
Similar to our method, the CB estimator computes prediction error using randomized train--test data pairs; however, it employs independent Gaussian randomization variables.
In short, our method relates to the CB method in a similar way as standard cross-validation relates to independent train--test splits using subsampled data.
Crucially, our method preserves a property of standard cross-validation where the original data can be reconstructed by pooling the train (or test) data from all folds, which is possible due to the choice of antithetic correlation between our randomization variables.

Another related proposal is the multi-fold data thinning (MDT) method introduced by \cite{neufeld2024data}, which also reconstructs the original data from the folds. 
However, the MDT method, like the sample-splitting based version of cross-validation, does not have the two distinct parameters of our approach that separately control bias and variance in the estimation procedure.
Moreover, it does not incorporate the same correlation structure among the randomization variables that our approach uses.

Finally, we point out that these existing approaches make different distributional assumptions on the data than our approach.
Unlike the CB and MDT methods, which propose randomization schemes for normal data or, in the case of MDT, for a larger family of parametric distributions, our approach uses a single antithetic Gaussian randomization scheme that can be applied across various data types, as long as the sufficient statistic in the loss function is approximately normal.

Our cross-validation proposal fits within a broader line of recent work on randomized alternatives to traditional sample splitting, using external noise to perform statistical tasks such as model validation, selective inference, and risk estimation. 
These approaches include splitting strategies such as data fission \citep{leiner2023data} and data thinning \citep{rasines2023splitting, dharamshi2025generalized}, methods leveraging Gaussian randomization for selective inference \citep{dai2023fdr, TianTaylor2018, PanigrahiTaylor2022, huang2023selective}, and randomized procedures for unbiased estimation of risk and prediction error \citep{tian2020prediction, oliveira2021unbiased, oliveira2022unbiased, fry2023unbiased}.

Here is a summary of the remainder of the paper:
\begin{enumerate}
\item In Section~\ref{sec:3}, we review the CB estimator for the normal means problem and introduce our cross-validated estimator using antithetic randomization variables. We conclude this section by establishing a connection between our cross-validated estimator and Stein's unbiased risk estimator (SURE), and showing that our estimator can be roughly viewed as SURE applied to a convolution-smoothed version of the prediction function.
\item In Section~\ref{sec: theory}, we analyze the MSE of the proposed estimator as $\alpha$, the parameter controlling bias, approaches zero. 
Our theory demonstrates that we can obtain unbiased estimates of prediction error as $\alpha \to 0$, while ensuring that the variance of our estimator remains stable even with vanishingly small $\alpha$, under mild smoothness assumptions.
\item In Section~\ref{sec:glm}, we extend our cross-validation approach to handle different data types within a relatively assumption-light framework, via randomization of the sufficient statistic in the loss function. Assuming the sufficient statistics are asymptotically normal and satisfy certain regularity conditions, we show that the MSE analysis in Section \ref{sec: theory} generalizes to a large category of loss functions, including those used in fitting generalized linear models (GLMs), such as logistic regression.
\item In Section~\ref{sec: experiments}, we provide simulation results comparing our method to standard cross-validation, the coupled bootstrap approach, and SURE. 
When sample splitting is difficult, our method offers a simple alternative, avoiding manual tuning of the bias--variance tradeoff and differentiating the prediction function.
\item Section \ref{sec:conclusion} concludes the paper with a brief summary of our contributions and identifies a few potential directions for future research.
\end{enumerate}

\section{The proposed estimator for the normal means problem}
\label{sec:3}

We first review the coupled bootstrap (CB) estimator proposed by \cite{oliveira2021unbiased}, before introducing our antithetic randomization scheme and our cross-validated estimator.

\subsection{Review of coupled bootstrap (CB)}

The CB estimator \citep{oliveira2021unbiased}  estimates the risk in the normal means problem, where the response vector $Y\in\R^n\sim \N(\theta,\sigma^2I_n)$, with a known variance $\sigma^2$. 
In our paper, we focus on estimating the prediction error, defined in \eqref{pred:error} as $\PE(g)= \EE{\|g(Y)- \tY\|_2^2}$, where $\tY \sim \N(\theta, \sigma^2 I_n)$ denotes an independent copy of $Y$.
In the normal means problem, the prediction error differs from the risk only by the additive constant $n\sigma^2$.

To estimate $\PE(g)$, the CB method generates $K$ independent Gaussian randomization variables:
$
\tilde\om^{(1)}, \tilde\om^{(2)}, \ldots, \tilde\om^{(K)}\iid \N(0, \sigma^2 I_n).
$
For each $k \in [K]$ and a parameter $\alpha \in \mathbb{R}^+$, two randomized copies of $Y$, which we call a train--test data pair, are constructed as
\begin{equation}
\label{CB:train:test}
\left(\tilde{Y}^{(k)}_{\text{train}}= Y + \sqrt{\alpha}\tilde\om^{(k)}, \quad \tilde{Y}^{(k)}_{\text{test}}=Y- \dfrac{1}{\sqrt{\alpha}}\tilde\om^{(k)}\right).
\end{equation}
By construction, the two vectors are distributed as
\[
\begin{pmatrix} \widetilde{Y}^{(k)}_{\text{train}} \\ \widetilde{Y}^{(k)}_{\text{test}}\end{pmatrix}  \sim \N\left(\begin{pmatrix}\theta \\ \theta \end{pmatrix}, \begin{bmatrix}\sigma^2 (1+\alpha) I_n & 0_{n, n} \\ 0_{n,n} & \sigma^2(1+\alpha^{-1})  I_n\end{bmatrix} \right).
\]
The prediction error using the $k$-th train--test pair is then computed as
\begin{equation*}
{\text{CB}}_{\alpha}^{(k)}= \big\|\tilde{Y}^{(k)}_{\text{test}} - g(\tilde{Y}^{(k)}_{\text{train}}) \big\|_2^2- \frac{1}{\alpha} \big\|\tilde\om^{(k)} \big\|_2^2,
\end{equation*}
where the second term, $\|\tilde\om^{(k)}\|_2^2/\alpha$, adjusts for the difference between the variance of the noisy test data and the variance of the original data $Y$.
Finally, the CB estimator is obtained by averaging over the $K$ independent draws of randomization and is given by:
\[
{\text{CB}}_{\alpha} = \frac{1}{K} \sum_{k=1}^K{\text{CB}}_{\alpha}^{(k)}.
\]

Since $\tY^{(k)}_{\text{train}}\sim\N(\theta,(1+\alpha)\sigma^2 I_n)$, straightforward calculations show that the CB estimator is unbiased for a noise-inflated version of the prediction error
\begin{align*}
    \PE_\alpha(g)=\EE{\|g(Y) - \tY\|_2^2 },\text{ where }Y\sim \N(\theta, (1+\alpha)\sigma^2 I_n ),\; \tY\sim \N(\theta,\sigma^2 I_n).
\end{align*}
This estimand represents the prediction error obtained when $g$ is trained on noisier data, with variance inflated by a factor of $(1+\alpha)$. 
As a result, $\CB_\alpha$ is biased for the actual prediction error $\PE(g)$, defined in Equation~\eqref{pred:error}. Note that the bias---the difference between the noise-inflated prediction error $\PE_{\alpha}(g)$ and the original estimand $\PE(g)$---converges to zero as the parameter $\alpha$ approaches zero. Nevertheless, similar to standard cross-validation, a bias--variance tradeoff arises here as well: reducing  bias by decreasing $\alpha$ comes at a cost of higher variance. 
\cite{oliveira2021unbiased} show that the variance of the CB estimator is of order $O((K\alpha)^{-1})$ as $\alpha$ approaches $0$, which means that for any fixed $K$, its variance diverges as the bias is driven to $0$.

We propose a novel antithetic randomization scheme, presented below, that overcomes this limitation of the CB estimator for a large class of smooth prediction functions.

\subsection{Antithetic randomization}

In our antithetic randomization scheme, we generate $K$ ($K\geq2$) randomization variables as follows:
\begin{equation}
    \om^{(1)},\ldots,\om^{(K)}\sim \N(0,\sigma^2 I_n), \text{ where } \text{Cov}(\om^{(j)},\om^{(k)})=-\frac{\sigma^2}{K-1}I_n \text{ for }j\neq k.
    \label{antithetic:rand}
\end{equation}

We make two important observations about the distribution of our randomization variables:
\begin{enumerate}
\item First, the normal distribution in \eqref{antithetic:rand} is degenerate.
This is because the variance of the sum of the randomization variables is a zero matrix, i.e., $\text{Var}\left(\sum_{k=1}^K \om^{(k)}\right)=0$. Combined with fact that the randomization variables have zero mean, this imposes the following zero-sum constraint on these randomization variables:
\begin{equation}
\sum_{k=1}^K \om^{(k)}=0.
\label{zero:sum}
\end{equation}
\item Second, for a $K$-by-$K$ correlation matrix where all off-diagonal entries are equal, the range of possible correlation values is $[-\frac{1}{K-1}, 1]$.
Therefore, we note that our randomization scheme takes the most negative correlation possible, which is why we term it an ``antithetic'' scheme.
\end{enumerate}

For a fixed $\alpha\in \mathbb{R}^+$, we construct randomized train--test copies of the data $Y$ as
\begin{align*}
\begin{pmatrix} Y^{(k)}_{\text{train}} \\ Y^{(k)}_{\text{test}} \end{pmatrix} = \begin{pmatrix} Y- \sqrt{\alpha}\displaystyle\sum_{j\neq k}\om^{(j)} \\  Y- \dfrac{1}{\sqrt{\alpha}}\om^{(k)} \end{pmatrix} = \begin{pmatrix} Y + \sqrt{\alpha}\om^{(k)} \\ Y- \dfrac{1}{\sqrt{\alpha}}\om^{(k)}\end{pmatrix},\;\text{ for } k\in[K],
\end{align*}
where the second equality is due to the zero-sum constraint in \eqref{zero:sum}. 
Notice that our construction of $K$ noisy train--test pairs mimics the standard $K$-fold cross-validation in the sense that when the train (or test) data from all $K$ folds are averaged, the randomization variables cancel, thereby recovering the original data $Y$.

Finally, our cross-validated estimator $\cv_\alpha$ is defined as
\begin{align}\label{equ: def cv}
    {\text{CV}}_{\alpha}= \frac{1}{K}\sum_{k=1}^K {\text{CV}}_{\alpha}^{(k)},
\end{align}
where ${\text{CV}}_{\alpha}^{(k)} = \|Y^{(k)}_{\text{test}} - g(Y^{(k)}_{\text{train}})\|_2^2- \dfrac{1}{\alpha}\|\om^{(k)}\|_2^2$.

While ${\text{CV}}_{\alpha}^{(k)}$ in our estimator has the same form as ${\text{CB}}_{\alpha}^{(k)}$ in the CB method, it uses the antithetic randomization scheme described above instead of independent randomization variables.
As we show in Section~\ref{sec: theory}, our correlated randomization scheme yields a substantial variance reduction for smooth predictions, ensuring that the variance of our cross-validated estimator remains bounded as $\alpha \to 0$, at which point its bias also vanishes.

\subsection{Connection with SURE}
\label{sec: SURE}

In the normal means problem, Stein's Unbiased Risk Estimator \citep[SURE;][]{stein1981estimation} is typically introduced as an unbiased estimator of the quadratic risk $\mathbb{E}\left[\|g(Y)-\theta\|_2^2\right]$ for weakly differentiable estimators $g$. Since the prediction error $\PE(g)$ defined in \eqref{pred:error} differs from the quadratic risk only by a constant $n\sigma^2$, we may equivalently define SURE as an unbiased estimator of prediction error:
\begin{align*}
    \mathrm{SURE}(g)= \|Y-g(Y)\|_2^2 + 2\sigma^2\nabla\cdot g(Y),
\end{align*}
where the divergence $\nabla\cdot g(Y)$ is the trace of the Jacobian of $g$. Under mild regularity conditions, $\mathrm{SURE}(g)$ is unbiased for $\PE(g)$.

Our antithetic cross-validation estimator, $\cv_\alpha$, is closely connected with SURE. Using the zero-sum property of the antithetic variates, $\sum_{k=1}^K \omega^{(k)} = 0$, our estimator can be rewritten as
\begin{align}\label{equ: cv decomp 2}
    \cv_\alpha = \frac1K\sum_{k=1}^K \|Y - g(Y+\sqrt\alpha\omega^{(k)})\|_2^2 +\frac2K\sum_{k=1}^K \frac{(\omega^{(k)})\tran \left(g(Y+\sqrt\alpha\omega^{(k)})-g(Y)\right)}{\sqrt\alpha}.
\end{align}
Thus, assuming $g$ is differentiable, our estimator has a well-defined limit as $\alpha\to 0$:
\begin{align}\label{equ:zero-alpha}
    \cv_{0^+} := \lim_{\alpha\to 0}\cv_\alpha = \|Y - g(Y)\|_2^2 +\frac2K\sum_{k=1}^K (\omega^{(k)})\tran \nabla g(Y) \omega^{(k)}.
\end{align}
Averaging over the randomness in $\omega^{(1)},\ldots,\omega^{(K)}$ (equivalently, letting $K\to\infty$) and applying the trace trick,
\[
\EE{\cv_{0^+}\mid Y} = \mathrm{SURE}(g).
\]
This result parallels \citet{efron2004estimation}, who connected standard cross-validation and SURE when $n$ is large. \citet{oliveira2021unbiased} established a similar connection for the CB estimator, which recovers SURE by first letting $K\to\infty$ and then $\alpha\to 0$; this order is necessary because CB is unstable for small $\alpha$. By contrast, our estimator has a well-defined finite-$K$ limit as $\alpha\to 0$, and the expectation of this limit coincides with SURE.

While the above convergence holds as $\alpha\to0$, we show next that for any fixed $\alpha > 0$, a noise-free version of our estimator (defined below) can be interpreted as SURE applied to a convolution-smoothed prediction function $g$.

Consider the expression for $\cv_\alpha$ in Equation~\eqref{equ: cv decomp 2}, and replace the term $g(Y+\sqrt\alpha\omega^{(k)})$ with its conditional expectation $\EE{g(Y+\sqrt\alpha\omega)\mid Y}$, where the expectation is over $\omega\sim\N(0,\sigma^2 I_n)$.
This leads to the noise-free version of our estimator:
\begin{align}
    \overline{\cv}_\alpha= \|Y - \EE{g(Y+\sqrt\alpha\omega)\mid Y }\|_2^2 + \frac{2}{\sqrt\alpha}\EE{\omega\tran g(Y+\sqrt\alpha\omega) \mid Y}.
    \label{noise:free:CV}
\end{align}
Equivalently, $\overline{\cv}_\alpha$ can be seen as a Rao-Blackwellized version of $\cv_\alpha$, with the randomness from $\omega^{(k)}$'s marginalized out. 
The next result states that the Rao-Blackwellized version $\overline{\cv}_\alpha$ of the proposed estimator coincides with the SURE when $g$ is replaced by its convolution-smoothed version $g^*:=g*\varphi_{\alpha\sigma^2}$. Here, $g*\varphi_{\alpha\sigma^2}(y):=\int g(y-z)\varphi_{\alpha\sigma^2}(z)\rd z$, where $\varphi_{\alpha\sigma^2}$ denotes the density of $\N(0, \alpha\sigma^2 I_n)$.
\begin{proposition}[Connection with SURE]{\label{prop: SURE}}
    It holds that
    \begin{align}\label{equ: smoothed cv}
        \overline{\cv}_\alpha = \mathrm{SURE}(g * \varphi_{\alpha\sigma^2} ).
    \end{align}
\end{proposition}
The proof is provided in Appendix~\ref{prf: prop SURE}.

To sum up, when SURE is available, our estimator $\cv_\alpha$ closely resembles SURE for small $\alpha$.
A computational advantage of $\cv_\alpha$ over SURE is that it avoids calculating the divergence term $\nabla \cdot g$, which may not be available in closed form for many prediction functions. Moreover, when SURE is not applicable---e.g., when $g$ is not weakly differentiable---$\cv_\alpha$ still remains well-defined. In such cases, it behaves as if SURE were applied to the infinitely differentiable, convolution-smoothed estimator $g * \varphi_{\alpha\sigma^2}$.

\section{Mean squared error analysis of our estimator}
\label{sec: theory}

Continuing with the normal means problem, in this section, we analyze the mean squared error (MSE) of the proposed estimator $\cv_\alpha$ in \eqref{equ: def cv} for estimating the prediction error $\PE(g)$, defined in \eqref{pred:error}.
Note, the MSE decomposes into bias and variance components as
\begin{align*}
\EE{(\cv_\alpha -\PE(g) )^2 } &= \left\{\EE{\cv_\alpha} -\PE(g) \right\}^2 + \Var{\cv_\alpha}\numberthis\label{equ: MSE decomposition} \\
&= \left\{\EE{\cv_\alpha} -\PE(g) \right\}^2 + \EE{\Var{\cv_\alpha\mid Y}} + \Var{\EE{\cv_\alpha\mid Y }}.
\end{align*}
Following \citet{oliveira2021unbiased}, we call $\EE{\Var{\cv_\alpha\mid Y}}$ and $\Var{\EE{\cv_\alpha \mid Y}}$ the \emph{reducible variance} and \emph{irreducible variance}, respectively, since (as we will see) the former can be made arbitrarily small by increasing $K$, whereas the latter cannot. 

We study the bias $(\EE{\cv_\alpha} -\PE(g))$ in Section~\ref{sec: bias}, followed by the reducible variance $\EE{\Var{\cv_\alpha\mid Y}}$ and the irreducible variance $\Var{\EE{\cv_\alpha\mid Y }}$ in Section~\ref{sec: variance}.
As will be evident from the following results, the bias and irreducible variance of $\cv_{\alpha}$ and the CB estimator, ${\text{CB}}_{\alpha}$, reviewed in Section \ref{sec:3}, are identical, with differences in behavior arising solely from their reducible variance.

\subsection{Bias}\label{sec: bias}
We show that the bias of our estimator can be made arbitrarily small as $\alpha$ approaches zero, under the mild condition that $g$ is square-integrable. 
This result follows directly from the ``approximation to the identity" property of the Gaussian density, as stated in Lemma \ref{lem: approximation to identity} below.

Let $\varphi_{\sigma^2}$ denote the density of $\N(0, \sigma^2 I_n)$ in dimension $n$; for simplicity, we omit the dependence on $n$, slightly abusing the notation.

\begin{lemma}[Approximation to the identity]
\label{lem: approximation to identity}
    Let $f$ be an integrable function under the Gaussian distribution $\N(\theta, \sigma^2 I_n)$. Then
    \begin{align*}
    f*\varphi_{\alpha\sigma^2}(Y)\stackrel{L_1}{\to} f(Y) \text{ as }\alpha\to 0.
    \end{align*}
\end{lemma}
\begin{proof} 
    This is a direct application of Lemma~\ref{lem: log p condition} and Lemma~\ref{lem: L1} in the Appendix.
\end{proof}
Lemma \ref{lem: approximation to identity} states that the convolution of a function with $\varphi_{\alpha\sigma^2}$ is close to the original function in the $L_1$ sense as $\alpha\to0$. 
In the context of our problem, this lemma implies that 
$\EE{g(Y+\sqrt\alpha\omega)\mid Y}\stackrel{L_1}{\to} g(Y)$
as $\alpha\to0$, which in turn allows us to conclude that the bias of our estimator converges to zero as $\alpha$ approaches zero.
This observation is formalized in the following theorem.

\begin{theorem}[Bias]\label{thm: bias}
    Assume that $\EE{\|g(Y)\|_2^2}<\infty$. Then we have
    \begin{align*}
        \lim_{\alpha\to0} \EE{\cv_\alpha } =\PE(g).
    \end{align*}
\end{theorem}
\begin{proof}[Proof of Theorem~\ref{thm: bias}]
    Since $\EE{\cv_\alpha}=\EE{\cv_\alpha^{(k)}}$, it is sufficient to compute the expectation of $\cv_\alpha^{(k)}$. 
    Observe that
    \begin{equation*}
    \begin{aligned}
        \EE{\cv_\alpha^{(k)}}&=\EE{\|Y-\frac{1}{\sqrt\alpha}\omega^{(k)} - g(Y+\sqrt\alpha\omega^{(k)})\|_2^2 - \frac{\|\omega^{(k)}\|_2^2}{\alpha} } \\
        &=\EE{\|g(Y+\sqrt\alpha\omega^{(k)})\|_2^2 - 2(Y-\frac{1}{\sqrt\alpha}\omega^{(k)})\tran g(Y+\sqrt\alpha\omega^{(k)}) }\\
        & \ \ \ \  + \EE{\|Y-\frac{1}{\sqrt\alpha}\omega^{(k)}\|_2^2}  - \EE{\frac{\|\omega^{(k)} \|_2^2}{\alpha}}\\
        &=\EE{\|g(Y+\sqrt\alpha\omega^{(k)})\|_2^2 } -2\EE{Y } \tran \EE{g(Y+\sqrt\alpha\omega^{(k)})}+ \EE{\|Y\|_2^2}  ,
        \end{aligned}
    \end{equation*}
    where we have used the facts that $Y+\sqrt\alpha\omega^{(k)} \indep Y-\frac{1}{\sqrt\alpha}\omega^{(k)}$, $Y\indep \omega^{(k)}$, and $\EE{\omega^{(k)}}=0$.
    Note that 
    \[
    \EE{\|g(Y+\sqrt\alpha\omega^{(k})\|_2^2 \mid Y } = \|g\|_2^2 * \varphi_{\alpha\sigma^2} (Y),
    \] 
    which converges in $L_1$ to $\|g(Y)\|_2^2$ as $\alpha\to0$, by Lemma~\ref{lem: approximation to identity}. 
    Similarly, applying Lemma~\ref{lem: approximation to identity} to the function $g_i(Y)$ for $1\leq i\leq n$ shows that that $\EE{g(Y+\sqrt\alpha\omega^{(k)})\mid Y }$ converges in $L_1$ to $g(Y)$. 
    
    This establishes that, as $\alpha\to0$,
    \begin{align*}
        \EE{\cv_\alpha^{(k)}} \to \EE{\|g(Y)\|_2^2} - 2\EE{Y}\tran \EE{g(Y)} + \EE{\|Y\|_2^2}.
    \end{align*}
    The right-hand-side equals $\PE(g)=\EE{\|\tilde Y-g(Y)\|_2^2 }$, where $\tilde Y$ is an independent copy of $Y$. This completes the proof.
 \end{proof}

More importantly, as we demonstrate next for a large class of smooth prediction functions, unlike the CB estimator, decreasing $\alpha$ does not increase the variance of our estimator.

\subsection{Variance reduction with antithetic randomization}
\label{sec: variance}

To analyze the variance of the proposed estimator $\cv_\alpha$, we impose a mild smoothness condition on the prediction function $g$.
This condition is the weak differentiability assumption underlying the construction of the classical SURE estimator~\citep{stein1981estimation}.
\begin{assumption}[Weak differentiability]\label{assump: weakly differentiable}
All components $g_i$ ($1\leq i\leq n$) of $g$ are weakly differentiable. That is, there exists a function $\nabla g_i:\R^n\to\R^n$, the weak derivative of $g_i$, such that
\begin{align*}
    g_i(y+z) - g_i(y) = \int_0^1 z\cdot \nabla g_i(y+tz)\rd t,
\end{align*} 
for almost all $y, z\in\R^n$. 
Denote the Jacobian matrix of $g$ as $\nabla g\in \R^{n\times n}$, where the $i$-th row is equal to $\nabla g_i$.
\end{assumption}

This class of functions includes many well-known estimators, such as ridge, lasso, group lasso, and generalized lasso, evaluated at fixed values of their penalty parameters; see, for example, the paper by \cite{tibshirani2012degrees}.

The following theorem provides the expression for the reducible variance of $\cv_\alpha$ as $\alpha$ approaches zero.

\begin{theorem}[Reducible variance]\label{thm: reducible variance}
    Suppose that Assumption~\ref{assump: weakly differentiable} holds. 
    \sloppy{Furthermore, let $\EE{\|g(Y)\|_2^4}<\infty$, $\EE{\|\nabla g(Y)\|_F^2}<\infty$.}
    Then, we have that
    \begin{align*}
        \lim_{\alpha\to0} \EE{\Var{\cv_\alpha\mid Y}}= \frac{4\sigma^4}{K-1}\EE{\|\nabla g(Y) \|_F^2 + \tr(\nabla g(Y)^2 )}.
    \end{align*}
\end{theorem}
\begin{remark}
The CB estimator, based on independent randomization variables, has reducible variance of order $O(1/(K\alpha))$, regardless of whether the prediction function is weakly differentiable or not. 
Theorem \ref{thm: reducible variance}, stated above, implies that the reducible variance of our cross-validated estimator remains bounded for any fixed $K>1$ as $\alpha\to0$, since it is free from $\alpha$ in the limit. 
This establishes a clear advantage of the proposed estimator over the CB estimator for weakly differentiable prediction functions $g$, whose variance diverges as $\alpha \to 0$ for any fixed $K$. See Remark \ref{rem:genfunctions} for a discussion of how the variance of our estimator compares with that of the CB estimator for a hard-thresholded function, which lies outside the class of weakly differentiable functions.
\end{remark}

We provide a sketch of the proof here to illustrate the role of antithetic randomization in achieving this reduction in variance, with the detailed proof deferred to Appendix~\ref{prf: thm reducible variance}. 
\begin{proof}[Proof sketch of Theorem~\ref{thm: reducible variance}]
    We first write
    \begin{align*}
        \cv_\alpha&=\frac1K\sum_{k=1}^K \|Y-\frac{1}{\sqrt\alpha}\omega^{(k)} - g(Y +\sqrt\alpha\omega^{(k)} )\|_2^2 - \frac{1}{\alpha}\|\omega^{(k)}\|_2^2\\
        &=\underbrace{\frac1K\sum_{k=1}^K  \|Y-g(Y+\sqrt\alpha\omega^{(k)})\|_2^2}_{(\Rom{1})} + 
        \underbrace{\frac1K\sum_{k=1}^K \frac{2}{\sqrt\alpha}\langle  \omega^{(k)} , g(Y+\sqrt\alpha\omega^{(k)})\rangle}_{(\Rom{2})} \numberthis\label{equ: CV decomp} \\
        &\qquad \qquad - \underbrace{\frac2K\sum_{k=1}^K \langle Y, \frac{1}{\sqrt\alpha} \omega^{(k)}  \rangle}_{=0} .
    \end{align*}
    Note that the last term is 0 because of the zero-sum property of the antithetic randomization variables, i.e., $\sum_{k=1}^K \omega^{(k)}=0$.
    Thus we have
    \[
    \Var{\cv_\alpha \mid Y} = \Var{(\Rom{1}) \mid Y} + \Var{(\Rom{2}) \mid Y} + 2 \cov[{(\Rom{1}), (\Rom{2})\mid Y}].
    \]
    For the first summation $(\Rom{1})$, we show that 
    $\Var{(\Rom{1}) \mid Y} \stackrel{L_1}{\to}  0$.
    This is because we can write this conditional variance as the convolution of an integrable function with the Gaussian density $\varphi_{\alpha\sigma^2}$, which converges in $L_1$ to 0 by Lemma~\ref{lem: approximation to identity}.

    For the second summation $(\Rom{2})$, we have by the definition of weak differentiability that
    \begin{align*}
        (\Rom{2}) 
        &=\frac{2}{K\sqrt\alpha } \sum_{k=1}^K \langle \omega^{(k)}, g(Y) + \int_0^1 \nabla g(Y+t\sqrt\alpha\omega^{(k)})\tran (\sqrt\alpha\omega^{(k)}) \rd t  \rangle\\
        &=\frac{2}{K}\sum_{k=1}^K {\omega^{(k)}}\tran \left[\int_0^1 \nabla g(Y+t\sqrt\alpha\omega^{(k)})\rd t\right]  \omega^{(k)}.\numberthis\label{equ: second term decomp}
    \end{align*}
    The last equality is again due to the fact that $\sum_{k=1}^K \omega^{(k)}=0$.
    The ``approximation to identity property" is applied again to show that 
    $$
    \Var{(\Rom{2}) \mid Y} \stackrel{L_1}{\to}  \Var{\frac{2}{K} \sum_{k=1}^K {\omega^{(k)}}\tran \nabla g(Y) \omega^{(k)}\mid Y }.
    $$
    The right-hand-side in the last display is the variance of a quadratic form of the Gaussian vector $(\omega^{(1)}, \ldots,\omega^{(K)})$, which has a closed form as given in the statement of the Theorem. 

    Lastly, $\cov[{(\Rom{1}), (\Rom{2})\mid Y}]\stackrel{L_1}{\to} 0$ by applying the Cauchy-Schwarz inequality:
    \begin{equation*}
        \begin{aligned}
        \EE{\cov[{(\Rom{1}), (\Rom{2})\mid Y}]} &\leq \EE{\sqrt{\Var{(\Rom{1}) \mid Y} \Var{(\Rom{2}) \mid Y}}}\\
        &\leq \sqrt{\EE{\Var{(\Rom{1}) \mid Y}}}\sqrt{\EE{\Var{(\Rom{2}) \mid Y}}}.
        \end{aligned}
    \end{equation*} 
\end{proof}

To complete the analysis of variance of our estimator, we provide the limit of the irreducible variance.
\begin{theorem}[Irreducible variance]\label{thm: irreducible variance}
    Under the same assumptions as in Theorem~\ref{thm: reducible variance}, we have that
    \begin{align*}
        \lim_{\alpha\to0}\Var{\EE{\cv_\alpha \mid Y }} = \Var{\|Y - g(Y)\|_2^2 + 2\sigma^2 \tr(\nabla g(Y)) }.
    \end{align*}
\end{theorem}
The proof is provided in Appendix~\ref{prf: irreducible}.

Combining the bias--variance results in Theorem \ref{thm: bias}, \ref{thm: reducible variance} and \ref{thm: irreducible variance}, we find that, as $\alpha\to0$,
\begin{align*}
\text{MSE}(\cv_{\alpha}) \to \Var{\|Y - g(Y)\|_2^2 + 2\sigma^2 \tr(\nabla g(Y)) } + \frac{4\sigma^4}{K-1}\EE{\|\nabla g(Y) \|_F^2 + \tr(\nabla g(Y)^2 )}.
\end{align*}

The following corollary highlights the infinite efficiency gain of our cross-validated estimator with antithetic randomization relative to the CB estimator with independent randomization: by selecting a small $\alpha$, we can make the bias arbitrarily small while ensuring that the variance of our estimator remains stable (i.e., the variance does not blow up).

\begin{corollary}
\label{cor:dominate CB}
Under the same assumptions as in Theorem~\ref{thm: reducible variance}, for any finite $K>1$, we have that
\begin{align*}
\lim_{\alpha \to 0} \left\{\mathrm{MSE}(\cv_{\alpha}) - \mathrm{MSE}(\mathrm{CB}_{\alpha})\right\} = -\infty. 
\end{align*}
\end{corollary}

Finally, we comment on the potential of our randomization scheme for a larger class of prediction functions than those analyzed here, leaving a precise characterization of its variance beyond weakly differentiable functions to future work.
\begin{remark}
In Appendix~\ref{sec: indicator}, we compute the variance of a hard-thresholded function with our antithetic randomization scheme. This estimator is discontinuous and therefore does not belong to the class of weakly differentiable functions. In particular, we show that for any fixed $K$, the reducible variance of our cross-validated estimator is of order $O\bigl(\frac{1}{K\sqrt\alpha})$, while that of coupled bootstrap estimator is $O(\frac{1}{K\alpha})$. 
As a result, the antithetic CV estimator continues to yield a substantial variance reduction relative to the CB estimator, with Corollary \ref{cor:dominate CB} still holding up for the hard-thresholded function. 
This theoretical finding is further confirmed by the accompanying simulations included in the same section.
\label{rem:genfunctions}
\end{remark}

\section{Generalization of our method with randomized sufficient statistics}
\label{sec:glm}

In this section, we move beyond the normal means problem and develop our cross-validation method to address the framework described in Section~\ref{sec:2} in full generality. This generalization accommodates different data types and loss functions, provided that the estimator derived from the loss function depends on a sufficient statistic that is asymptotically normal.

Fixing some notations for this section, suppose the data $Y=Y_n$ is generated from an exponential family with density function:
\begin{equation*}
    p_n(Y_n \mid \theta_n) =  \exp\left\{\sqrt{n}(\theta_n\tran S_n(Y_n) - A_n(\theta_n))\right\}\cdot h_n(Y_n),
    \label{gen:density}
\end{equation*}
where $\theta_n$ is the $p$-dimensional natural parameter. 
Throughout, we use subscripts $n$ on the data, estimand, and estimator to emphasize their dependence on the sample size.

Note, in the above formulation, the sufficient statistic $S_n(Y_n)=S_n$ and the log-partition function $A_n(\theta_n)$ are both scaled by $1/\sqrt n$.
To estimate the unknown parameter $\theta_n$, it is common to use a loss function derived from the negative log-likelihood, given by
\begin{equation}
    \calL(\theta_n, Y_n)=  A_n(\theta_n)-\theta_n\tran S_n(Y_n) - \frac{1}{\sqrt n}\log h_n(Y_n) ,
\label{gen:loss}
\end{equation}
which depends on $Y_n$ only through the sufficient statistic $S_n$.

As discussed earlier, we assume the sufficient statistic is approximately normal.
Formally, we posit the existence of a sequence of $p \times p$ positive definite matrices $H_n$ and vectors $\mu_n \in \mathbb{R}^p$ such that
\begin{equation}
    H_n^{-1/2}(S_n-\mu_n) \stackrel{d}{\Rightarrow} \N(0, I_p).
    \label{asymptotic:normal:stats}
\end{equation}

This problem formulation accommodates commonly used loss functions, including those employed in fitting generalized linear models (GLMs). Moreover, in GLMs, the asymptotic normality of $S_n$ holds under standard regularity conditions, as established in \cite{fahrmeir1985consistency}.
In what follows, we first present our method for creating train--test pairs by randomizing the sufficient statistic $S_n$, and then analyze the bias and variance of the resulting cross-validated estimator in the asymptotic regime as $n \to \infty$.

\begin{remark}
The exponential family framework described here provides concrete examples of loss functions that depend on approximately normal sufficient statistics, though our method would also apply to other loss functions with such statistics.
\end{remark}

\subsection{Cross-validated estimator}

Let $g(S_n)$ be our estimator that depends on the sufficient statistic $S_n$.
As before, we define the prediction error as
\begin{align*}
    \mathrm{PE}_n(g)=\EE{\calL(g(S_n), \tilde Y_n ) }= \EE{A_n(g(S_n)) - g(S_n)\tran \tilde{S}_n - n^{-1/2}\log h_n(\tY_n)},
\end{align*}
where $\tilde Y_n$ is an independent copy of $Y$, and $\tilde{S}_n= S_n(\tilde{Y}_n)$ is the sufficient statistic of $\tilde{Y}_n$.
Compared with \eqref{pred:error}, where the prediction performance is assessed with respect to the quadratic loss, the loss function considered here is based on the negative log-likelihood of the data.

We define a rescaled version of our sufficient statistics as:
\[
T_n = H_n^{-1/2} S_n, \quad \tilde T_n=H_n^{-1/2} \tilde{S}_n.
\]
By Equation~\eqref{asymptotic:normal:stats},
the asymptotic distributions of $T_n-H_n^{-1/2}\mu_n$ and $\tilde T_n-H_n^{-1/2}\mu_n$ are $\N(0, I_p)$.
Furthermore, let
$
\mathfrak{g}_n(T_n)= (H_n^{1/2})\tran g(H_n^{1/2} T_n)$, $\mathfrak{A}_n(T_n)= A_n(g(H_n^{1/2}T_n))
$.
As per these notations, we have that
\[
A_n(g(S_n))=\mathfrak{A}_n(T_n),\quad g(S_n)\tran \tilde S_n=\mathfrak g(T_n)\tran \tilde T_n.
\]
Using these notations, we can rewrite the prediction error as
\begin{equation}
    \mathrm{PE}_n(g)=\EE{\mathfrak{A}_n(T_n) - \mathfrak{g}_n(T_n) \tran \tilde T_n} -\EE{n^{-1/2}\log h_n(\tilde Y_n)}.
\label{PE:general}
\end{equation}
Note that the second expectation in our estimand, $\EE{n^{-1/2}\log h_n(\tilde Y_n)}$, can be easily estimated by $n^{-1/2}\log h_n(Y_n)$. 
The first expectation requires an appropriate estimator. However, it is taken over $T_n$ and $\tilde T_n$, which are asymptotically normal vectors with identity covariance.
Consequently, the problem reduces to a form analogous to the normal means problem discussed earlier, except that $T_n$ is not exactly normal but asymptotically normal. 

To proceed, we apply the same idea as before, constructing the train--test pair of randomized data by adding noise to the rescaled sufficient statistic, as
\begin{align*}
    \left(T_n + \sqrt\alpha\omega, \quad T_n-\frac{1}{\sqrt\alpha} \omega\right), \quad \text{where } \omega\sim \N(0, I_p),
\end{align*}
for $\alpha \in \mathbb{R}^+$.
Clearly, the train--test data vectors are asymptotically independent. 
We fit our estimator on $T_n+\sqrt\alpha\omega $ and evaluate its predictive performance on $T_n-\frac{1}{\sqrt\alpha}\omega$, which yields the following estimate of $\PE_n(g)$:
\begin{align*}
    \frakA_n(T_n+\sqrt\alpha\omega) - \frakg_n(T_n + \sqrt\alpha\omega )\tran (T_n - \frac{1}{\sqrt\alpha}\omega) - n^{-1/2}\log h_n(Y_n).
\end{align*}
We repeat this procedure $K>1$ times, using randomization variables $\omega^{(1)}, \ldots, \omega^{(K)}$ generated via the antithetic scheme in \eqref{antithetic:rand}.
Averaging over the $K$ replicates, we obtain our cross-validated estimator
\begin{equation}
\begin{aligned}
    \cv_{n,\alpha}=\frac1{K}\sum_{k=1}^K&\Big\{\mathfrak A_n( T_n+\sqrt\alpha\omega^{(k)}) - \mathfrak g_n(T_n + \sqrt\alpha\omega^{(k)} )\tran (T_n - \frac{1}{\sqrt\alpha}\omega^{(k)})  \Big\}\\
    &\quad - n^{-1/2} \log h_n(Y_n).
\end{aligned}    
\label{CV:general}
\end{equation}

Before analyzing the bias--variance properties of our estimator in \eqref{CV:general}, we first make a few remarks.

\begin{remark}
First, we observe that it is possible to work directly with the sufficient statistics $S_n$ without rescaling them to $T_n$.
In this case, the introduced randomization variables must have a marginal covariance matrix equal to $H_n$, while preserving the same antithetic correlation structure used throughout. We apply rescaling  here only to simplify the presentation of subsequent theory.
\end{remark}
\begin{remark}
Second, we note that our approach of randomizing the sufficient statistic $S_n$ within the loss function bears close resemblance to many existing approaches in the selective inference literature. The idea of splitting sufficient statistics using external randomization, which the authors refer to as thinning, was described in \cite{dharamshi2025generalized}. Instead of assuming parametric distributions for the sufficient statistics, our method exploits their asymptotic normality, without directly imposing parametric assumptions on the data. 
For example, using this idea, we provide a convenient way to construct train--test pairs by randomizing the asymptotically normal sufficient statistic in examples such as logistic regression.
\end{remark}
\begin{remark}
At last, we observe that our scheme of randomizing the sufficient statistic is equivalent to adding a linear term both in the parameter and the randomization variable to the original loss function. This type of randomization has been used to perform selective inference after penalized regression, such as the lasso or group lasso in \cite{PanigrahiTaylor2022, panigrahi2024exact, panigrahi2023approximate}. \cite{liu2025selective} employed the same antithetic approach as this work to conduct selective inference on distributed data for GLMs, framing the selection problem as an asymptotically equivalent version of the lasso problem with antithetic Gaussian randomization.
\end{remark}

As we demonstrate next, our estimator exhibits bias--variance properties similar to those in the normal means problem: the asymptotic bias of our estimator vanishes as $\alpha \to 0$, while its variance remains bounded.

\subsection{Mean squared error analysis}

To conduct the mean squared error analysis of our cross-validated estimator $\cv_{n,\alpha}$, we require some additional assumptions on the sufficient statistic $T_n$.

For a weakly differentiable $\mathbb{R}^p$-valued function $g$ and a $p$-dimensional vector $\mu$, define Stein's operator as
\begin{align*}
(\calT_{\mu} g)(x)=\langle g(x),\mu-x \rangle + \nabla\cdot g(x).
\end{align*}
For a normal random variable $X\sim \mathcal{N}(\mu, I_p)$, it follows that $\EE{(\calT_\mu g)(X) }=0$, which recovers Stein's identity.

Let $\mathbb{Q}_n$ represent the distribution of the rescaled sufficient statistics, $T_n$, with density $q_n$ and expectation $m_n= H_n^{-1/2}\mu_n$.

\begin{assumption}\label{assump: stein discrepancy}
Assume that
    \begin{align*}
        \lim_{n\to\infty}\EE{(\calT_{m_n} g_n) (T_n) } = 0
    \end{align*}
where
$
(\calT_{m_n} g)(x)= \langle g(x), m_n-x\rangle + \nabla\cdot g(x).
$
\end{assumption}

Under a distribution $\mathbb{Q}_n$ that is not normal, note that $\EE{(\calT_{m_n} g_n) (T_n) }$ is no longer exactly zero.
This quantity, known as Stein's measure of non-normality, forms the basis for the notion of Stein's discrepancy; see, for example, the paper by \cite{gorham2015measuring}.
Assumption \ref{assump: stein discrepancy} requires that the sufficient statistics exhibit vanishingly small Stein's discrepancy as $n$ goes to infinity.
For example, given that the sufficient statistics are asymptotically normal, this condition holds if $\|T_n\|_q^q$ is also uniformly integrable, and both functions $\langle g(x), x\rangle$, $\nabla\cdot g(x)$ grow slower than $\|x\|_q^q$ for some $q>0$.

\begin{assumption}\label{assump: log density q_n}
    Assume that there exist constants $N_0>0$ and $C>0$ such that, for all $n\geq N_0$, the density $q_n$ of $T_n$ satisfies
    $|\log q_n(x) -\log q_n(x')| \leq C \|x-x'\|_2^2.$
\end{assumption}

The condition in Assumption \ref{assump: log density q_n} is automatically satisfied if the density of the sufficient statistics converges to a normal density.

Now we are ready to show that the bias and variance results established in Section~\ref{sec: theory} for exactly normal data carry over to our estimator based on asymptotically normal sufficient statistics.

\begin{theorem}[Bias]\label{thm: glm bias}
    Let Assumptions~\ref{assump: weakly differentiable}, \ref{assump: stein discrepancy}, and \ref{assump: log density q_n} hold. 
    In addition, assume that \sloppy{$\EE{|\frakA_n(T_n)|}<\infty$}, $\EE{\|\frakg_n(T_n)\|_2^2}<\infty$, and $\EE{\|\nabla\frakg_n(T_n)\|_F}<\infty$. 
    Then
    \begin{align*}
        \lim_{n\to\infty} \lim_{\alpha\to0} \Big|\EE{\cv_{n,\alpha}} - \PE_n(g)\Big| = 0.
    \end{align*}
\end{theorem}

\begin{theorem}[Reducible variance]\label{thm: glm var}
    Let Assumptions~\ref{assump: weakly differentiable} and \ref{assump: log density q_n} hold. 
    In addition, assume that $\EE{\frakA_n(T_n)^2}<\infty$, $\EE{\|\frakg_n(T_n)\|_2^4}<\infty$, and $\EE{\|\nabla\frakg_n(T_n)\|_F^2}<\infty$. 
    When $n\geq N_0$, we have
    \begin{align*}
        \lim_{\alpha\to0} \EE{\Var{\cv_{n,\alpha} \mid Y_n }}=\frac{1}{K-1}\EE{\|\frakg_n(T_n)\|_F^2 + \tr(\nabla\frakg_n(T_n)^2) }.
    \end{align*}    
\end{theorem}
The proofs are provided in Appendix~\ref{prf: glm bias} and \ref{prf: glm var}.

As a concrete example within our framework, we consider logistic regression, a widely used method for classification.

\begin{example}[Logistic regression] 
  In logistic regression, the negative log-likelihood is given by 
  $\sum_{i=1}^n - y_i x_i\tran \theta + \log(1+e^{x_i\tran\theta}),$
  \sloppy{where the data follows the model $y_i\sim \mathrm{Bernoulli}(\pi(x_i\tran \theta_n) )$ and $\pi(x)=\frac{1}{1+\exp(-x)}$ is the sigmoid function.}
  To fit our setup, we scale the log-likelihood by $1/\sqrt{n}$, which yields the loss function
  $\calL(\theta_n, Y_n)= -\theta_n\tran S_n + A_n(\theta_n)$,
  where 
  $S_n=\frac{1}{\sqrt n} X_n\tran Y_n$, $A_n(\theta_n)= \frac{1}{\sqrt{n}}\sum_{i=1}^n\log(1+e^{x_i^\top \theta_n})$.

  As noted earlier, under the mild regularity conditions in \cite{fahrmeir1985consistency}, $S_n$ is asymptotically normal, with covariance matrix $H_n=\dfrac{1}{n}X_n\tran W_n X_n$, where $W_n$ is a diagonal matrix with entries $\pi(x_i\tran \theta_n)(1-\pi(x_i\tran \theta_n))$, for $i\in [n]$. Even though this matrix depends on the unknown parameter, our method performs well as long as this matrix can be accurately estimated, as demonstrated by our simulations in Section \ref{sec: experiments}.

  We note that our approach does not apply to arbitrary binary classification problem, but only to those in which the prediction error can be viewed as depending on an approximately normal sufficient statistic.
\end{example}

\section{Numerical experiments}
\label{sec: experiments}

In this section, we evaluate our method empirically on simulated data and compare our method with competing estimators for estimating prediction errors. Additional numerical results can be found in Appendix~\ref{sec: lasso}. The code to reproduce the experiments is available at \sloppy{\href{https://github.com/liusf15/Antithetic-CV}{github.com/liusf15/Antithetic-CV}. }

\subsection{Effects of $\alpha$ and $K$}

We start by examining the effects of $\alpha$ and $K$, the two parameters that control the bias and variance of the proposed cross-validation estimator, respectively.

We consider an isotonic regression problem with $n=100$ observations, in which the one-dimensional predictors are drawn independently from $\textnormal{Uniform}(0,1)$ and are treated as fixed throughout this experiment. 
The responses are then generated independently as $Y_i \sim \N(f^*(X_i), \sigma^2)$, where $f^*(x) = 2\cdot\lceil 5x\rceil-6$ and $\sigma^2 = 1$.
Visualizations of $f^*$ together with the simulated data points are shown in Figure~\ref{fig:true-monotone-function}.
We apply isotonic regression to estimate the true prediction function $f^*(x)$.
A snapshot of the same example was presented in the introduction to showcase comparisons between our method and the standard CV method.

\begin{figure}[t]
    \centering
    \includegraphics[width=.4\textwidth]{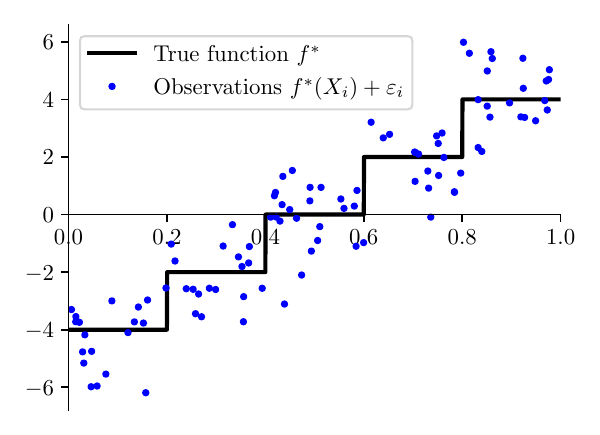}
    \caption{In the isotonic regression simulations, the data are generated based on the function $f^*$, shown in solid line. An example of a simulated dataset is displayed as scatter points.}
    \label{fig:true-monotone-function}
\end{figure}

We compare the MSE of the proposed ``Antithetic CV" method and the ``Coupled Bootstrap" (CB) method across various values of $\alpha$ and $K$. The error bars represent variation over 10000 independent simulations. In the left panel of Figure~\ref{fig: isotonic alpha K}, $\alpha$ is fixed at 0.05, and $K$ is varied from 2 to 32. In the right panel, $K$ is fixed at 16, and $\alpha$ is varied from 0.005 to 0.5. In all comparisons, standard CV uses the same value of $K$ for train--test sample splitting, as the other two methods. We include standard CV in our comparisons to illustrate how its estimation quality for the prediction error compares to that of the other two estimators on the same problem. We include SURE as a benchmark, which is applicable here since the isotonic regression estimator is piecewise linear.

From Figure~\ref{fig: isotonic alpha K}, we make a few key observations: 
\begin{enumerate}
    \item Our antithetic CV method consistently outperforms standard $K$-fold CV and the coupled bootstrap method across all configurations of $(\alpha, K)$. It achieves a similar MSE as SURE when $K=16$ and $\alpha\leq 0.1$, without explicitly using the divergence of the isotonic regression estimator.
    
    \item In the left panel of Figure~\ref{fig: isotonic alpha K}, the MSE decreases for all three methods as the number of train--test repetitions $K$ increases. For the coupled bootstrap and antithetic CV methods, this reduction can be directly attributed to the decrease in reducible variance with increasing $K$.
    
    \item In the right panel of Figure~\ref{fig: isotonic alpha K}, we observe that the MSE of CB increases as $\alpha$ decreases. This is because the MSE of CB is dominated by the variance when $\alpha$ is small. As a result, even though the bias of CB decreases with smaller $\alpha$, the increase in variance outweighs this benefit, leading to a higher overall MSE. In contrast, for our antithetic CV method, we note that the MSE either decreases or remains nearly unchanged as $\alpha$ becomes smaller. This is expected as our approach achieves bias reduction without paying a price in terms of the variance.
    
\end{enumerate}

Based on consistent evidence from our simulations, we summarize the following takeaways for the choice of $\alpha$ and $K$ to guide practical use.
\begin{remark}\emph{Choice of $\alpha$.}
When the estimator is differentiable, there is no price to choosing a very small $\alpha$, which is also supported by our theory. We also point out that our method, for positive $\alpha$, does not require explicitly differentiating the estimator $g$, unlike SURE. 
In our experiments, we found $\alpha=0.1$ performed essentially as well as smaller values.
\end{remark}

\begin{remark}\emph{Choice of $K$.}
For the choice of $K$, there is a straightforward computational--statistical tradeoff: from a statistical perspective, larger values of $K$ are preferred, whereas computational costs scale linearly with $K$. In our experiments, we found that $K$ around $10$ yields results comparable to those obtained with much larger values of $K$.
\end{remark}

\begin{figure}
  \centering
  \begin{subfigure}{.45\textwidth}
      \includegraphics[width=\textwidth]{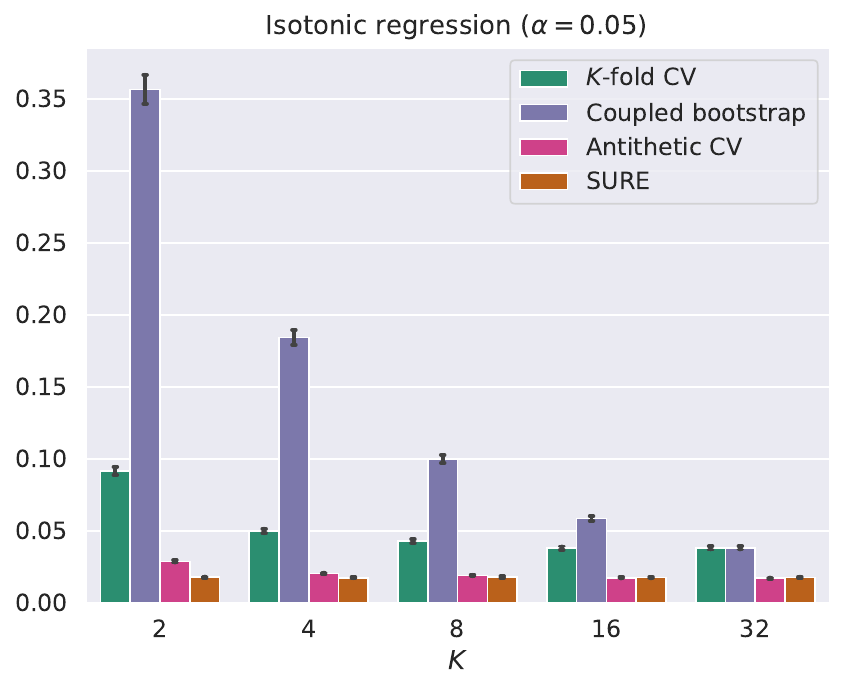}
  \end{subfigure}
  \begin{subfigure}{.45\textwidth}
      \includegraphics[width=\textwidth]{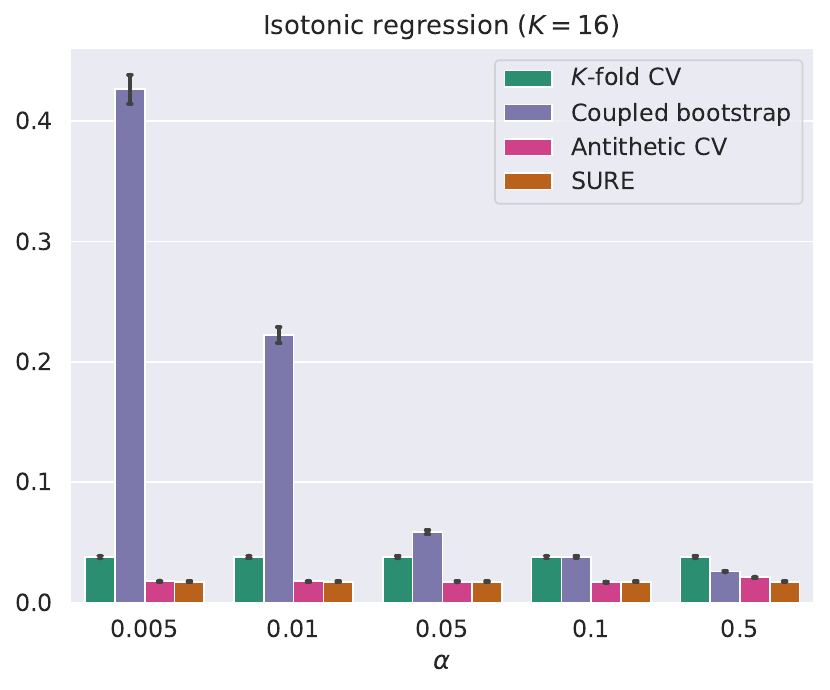}
  \end{subfigure}
  \caption{Mean squared error of standard $K$-fold CV, coupled bootstrap with $K$ repetitions, and the proposed Antithetic CV methods with $K$ repetitions. Left panel: MSE with $\alpha=0.05$ and varying $K$. Right panel: MSE with $K=16$ and varying $\alpha$. }
  \label{fig: isotonic alpha K}
\end{figure}

\subsection{Logistic regression}

We illustrate the effectiveness of the proposed method for estimating prediction errors in GLMs, as described in Section~\ref{sec:glm}. 

Below, we consider a logistic regression example with both continuous and categorical predictors and $n=100$. 
The $4$ continuous features $X^{\text{cts}}$ are generated from the standard Gaussian distribution independently, and  each of the $2$ categorical features, $X^{\text{cat}}$ has 3 classes drawn with probability $[0.1, 0.1, 0.8]$ independently. 
Given the predictors, the response $Y_i$ is drawn from a Bernoulli distribution with mean parameter $(1+e^{-\eta_i})^{-1}$, where
$\eta_i =\beta\tran X_i^{\text{cts}} + \sum_{j=1}^{2}\sum_{k=1}^{3} \gamma_{jk}\mathbf{1}\{X_{ij}^{\text{cat}} =k\}$,
with $\beta = (1, -1, 1, -1)$ and $\gamma_{j} = (1/2, -1/2, 0)$ for $j=1,2$.  The regression coefficients are estimated by fitting a logistic regression model, where we have used one-hot encoding for the categorical features. 
Our goal is to estimate the prediction error, defined here as the expectation of the negative log-likelihood under the logistic regression model.

Our antithetic CV method assumes knowledge of the covariance matrix of the sufficient statistic or requires an estimator for this unknown matrix. Here, we demonstrate empirically that the proposed method continues to perform quite well even when we use a plug-in estimate for the covariance matrix. 
Specifically, in our experiment, we estimate the unknown covariance matrix of the sufficient statistic $H_n$ by $\hat H_n=\frac1n X\tran \hat W X$, where $\hat W$ is the diagonal matrix whose $i$-th diagonal entry is $\pi(x_i\tran \hat\theta ) (1-\pi(x_i\tran \hat\theta))$ and where $\hat\theta$ is the  logistic regression solution. We then use the antithetic CV method introduced in Section~\ref{sec:glm} with $H_n$ replaced by $\hat H_n$.

As a baseline for comparison, we implement an alternative method that substitutes the antithetic Gaussian variables in our approach with independent Gaussian variables. This method mirrors the CB method in the normal means problem, and is referred to as ``independent randomization.'' We set $\alpha=0.1$ for both the antithetic CV method and the method based on a Gaussian randomization scheme with independent variables. In addition, similar to our previous comparisons, we apply the standard $K$-fold CV estimator in this example. We consider two commonly used values, $K=10$ and $K=20$, which are standard choices in practice. The simulation is repeated 10000 times.

The MSE of all three estimators is presented in Figure~\ref{fig: logistic and mlp} (left panel). Our antithetic CV method achieves a much smaller MSE compared to the standard $K$-fold CV. Furthermore, we observe again a clear advantage of the proposed antithetic Gaussian randomization over an independent randomization scheme. Moreover, these results show that the antithetic CV method remains effective even when applied to sufficient statistics that are only asymptotically Gaussian and when the covariance matrix is unknown.

\begin{figure}
    \centering
    \begin{subfigure}{.46\textwidth}
        \includegraphics[width=\textwidth]{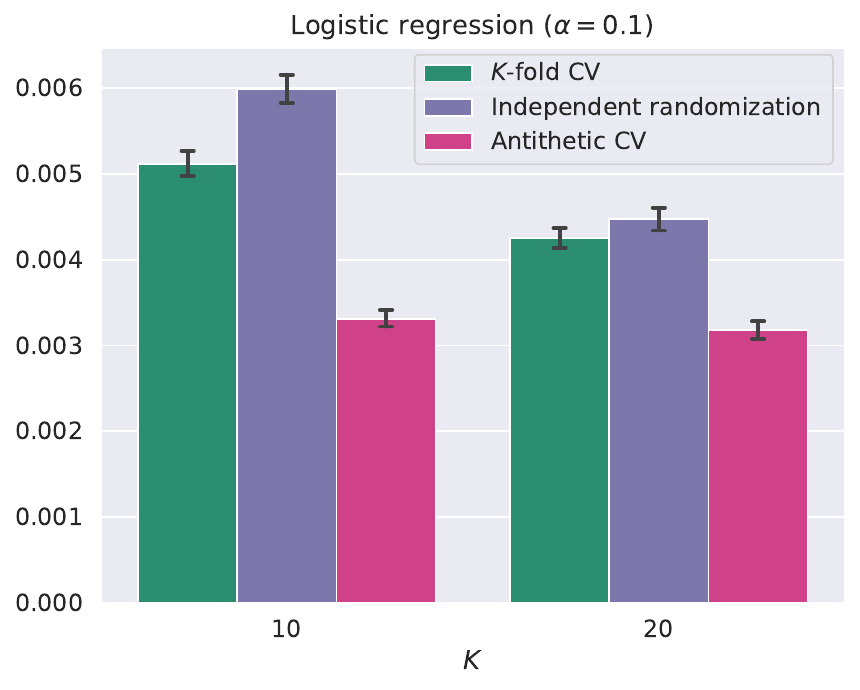}        
    \label{fig: logistic vary K}
    \end{subfigure}\hfill
    \begin{subfigure}{.45\textwidth}
        \includegraphics[width=\textwidth]{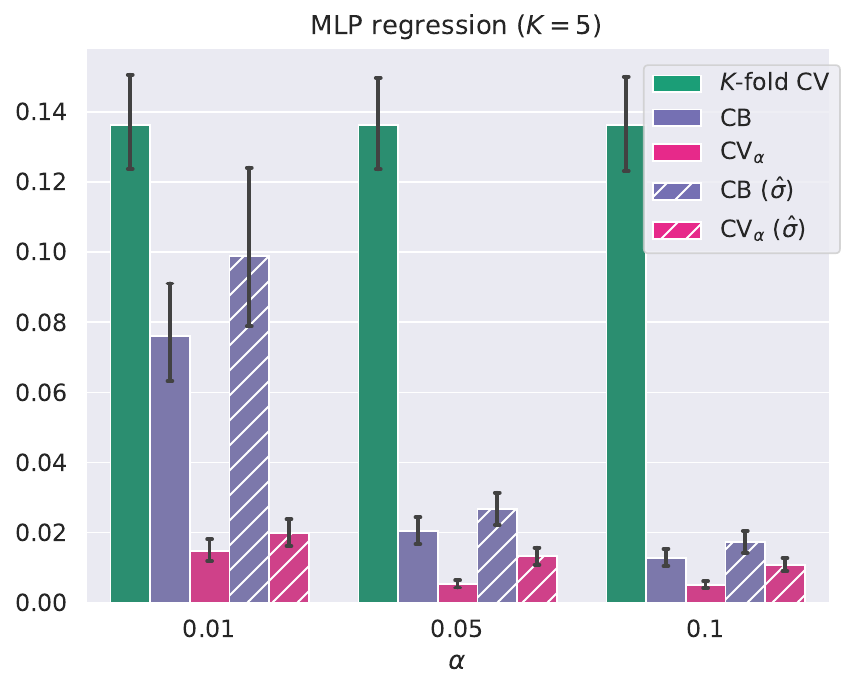}        
    \label{fig: mlp vary K}
    \end{subfigure}\vfill
    \caption{Left panel: MSE in estimating PE in the logistic regression example, with $\alpha=0.1$ and $K\in\{10,20\}$. Right panel: MSE in estimating PE in the MLP regression example, with $K=5$ and $\alpha\in\{0.01,0.05,0.1\}$. Striped bars correspond to CB and $\cv_\alpha$ with a plug-in estimate $\hat\sigma$.}
    \label{fig: logistic and mlp}
\end{figure}

\subsection{Multi-layer perceptron regression}

In this example, we apply the antithetic CV method to estimate the prediction error of a black-box neural network. 
The responses are drawn independently as $y_i\sim \N(f^*(x_i) , \sigma^2)$, where $f^*$ is the Friedman \#1 function \citep{friedman1991multivariate}, defined as:
$f^*(x) = 10\sin(\pi x_1x_2) + 20(x_3-\frac12)^2 + 10x_4 + 5x_5$,
and $\sigma^2=1$.
We generate $10$-dimensional predictors $x_i$ independently and uniformly from $[0,1]^{10}$. 
We fix $n=1000$. A multi-layer perceptron (MLP) prediction model is fitted using the Python package scikit-learn \citep{scikit-learn}. The MLP consists of two hidden layers, each with 64 units. The training configuration includes a maximum of $2000$ epochs, a batch size of $64$, and an $L_2$ regularization parameter of $0.01$, with all other parameters set to their default values. The simulation is repeated 200 times.

We evaluate the proposed antithetic CV method alongside the CB method, using either the true $\sigma$ or an estimated value $\hat\sigma$, where $\hat\sigma$ is taken as the standard error of the residuals from the fitted MLP model. The plug-in versions are denoted ``CB ($\hat\sigma$)'' and ``$\cv_\alpha$ ($\hat\sigma$).''
We fix $K = 5$ and vary $\alpha \in \{0.01, 0.05, 0.1\}$.
As shown in Figure~\ref{fig: logistic and mlp} (right panel), the proposed method achieves a smaller MSE than both the standard $K$-fold CV and the CB method. While the use of the plug-in estimate $\hat\sigma$ slightly degrades the performance of both CB and $\cv_\alpha$, the effect is not substantial.

\subsection{Hyperparameter tuning in fused lasso}

To illustrate the effectiveness of our antithetic CV method in hyperparameter tuning, we apply it to the fused lasso algorithm for image denoising, a setting previously examined by \citet{oliveira2021unbiased}. 
Given a noisy image $Y\in\R^{d\times d}$, the fused lasso estimator is defined as
\begin{align*}
    g_{\lambda}(Y) = \argmin_{\theta\in \R^{d\times d} } \frac12\|Y - \theta\|_F^2 + \lambda \Big(\sum_{\substack{1\leq i\leq d-1, \\ 1\leq j\leq d}} |\theta_{i+1,j} - \theta_{i,j} | + \sum_{\substack{1\leq i\leq d, \\ 1\leq j\leq d-1}} |\theta_{i,j+1} - \theta_{i,j} |  \Big) ,
\end{align*}
where the regularization term imposes a penalty on differences between adjacent pixel values. 
Our goal is to pick the regularization parameter $\lambda$ that minimizes the risk \sloppy{$\EE{\|g_\lambda(Y) - \theta\|_F^2}$}, which, up to the additive constant $d^2\sigma^2$, coincides with  the prediction error.

In this experiment, we search over a grid of 20 values of $\lambda$, equally spaced on a logarithmic scale from 0.04 to 4. The true risk is estimated over 100 simulation rounds given the clean image, yielding an optimal $\lambda=0.28$. The noise level is set to $\sigma^2=0.07$, and the image size is $398\times 398$.

First, we compare antithetic CV with coupled bootstrap (CB) across different values of $\alpha$, measuring how often each method selects the optimal $\lambda$ over 100 simulations. 
Note that the number of times the true $\lambda$ is selected serves as a proxy for how often a correct model is chosen.
To keep computational cost manageable, both CV and CB are implemented with $K=2$, and the results are summarized in Table~\ref{tab: fused lasso}.
\vspace{1em}
\begin{center}
    \begin{tabular}{lrrr}
    \toprule
    \hline
    $\alpha$ & 0.01 & 0.1 & 0.2 \\
    \hline
    \midrule
    CB & 0.64 & 0.49 & 0.00 \\
    CV & 0.96 & 0.51 & 0.00 \\
    \bottomrule
    \end{tabular}
    \captionof{table}{Proportion of times CB and CV select the optimal $\lambda$ over 100 independent simulations, for $\alpha \in \{0.01, 0.1, 0.2\}$. }
    \label{tab: fused lasso}
\end{center} 
We also compute the reconstruction error of CB and $\cv_\alpha$ at $\alpha=0.01$. Relative to the method that selects $\lambda$ based on the oracle risk, CB incurs an average excess reconstruction error of $0.02$, whereas $\cv_\alpha$ achieves a much smaller excess error of 0.001.

Next, in Figure~\ref{fig: fused lasso} (left), we show the estimated risk curves for each method compared with the ground truth, with shaded regions indicating variability across simulations.
Finally, in Figure~\ref{fig: fused lasso} (right), we present an example of a noisy image $Y$ and its denoised version $g_\lambda(Y)$ at the optimal $\lambda$ picked by our method.
Several remarks are in order:
\begin{itemize}
    \item When $\alpha=0.01$, the average estimated risks of both CV and CB align closely with the true risk. However, as shown in Figure~\ref{fig: fused lasso} (left), CB exhibits much higher variability, as evidenced by the wide shaded region, whereas the variability of CV is almost invisible. As a result, CV selects the optimal $\lambda$ in $96$ of $100$ runs, compared to only $64$ for CB.
    \item When $\alpha=0.2$, both methods incur non-negligible bias. Importantly, this bias varies with $\lambda$, causing both CV and CB to consistently choose suboptimal values in all simulations. These results highlight the importance of choosing a small $\alpha$ to reduce the bias for hyperparameter tuning and selecting a correct model, and the importance of using antithetic CV to keep the variance low while using small $\alpha$.
\end{itemize}
Finally, we emphasize that within each simulation and each method, the same randomization variables are shared across different values of $\lambda$. If one instead uses different randomization variables for different $\lambda$, the estimated risk curves of CB become much noisier, and the frequency of optimal selections of CB drops below 10\%. This is because the estimated risks at different $\lambda$ are highly correlated when using the same randomization variables, which helps stabilize the selection of $\lambda$.

\begin{figure}
    \centering
    \begin{subfigure}{.5\textwidth}
    \includegraphics[width=\textwidth]{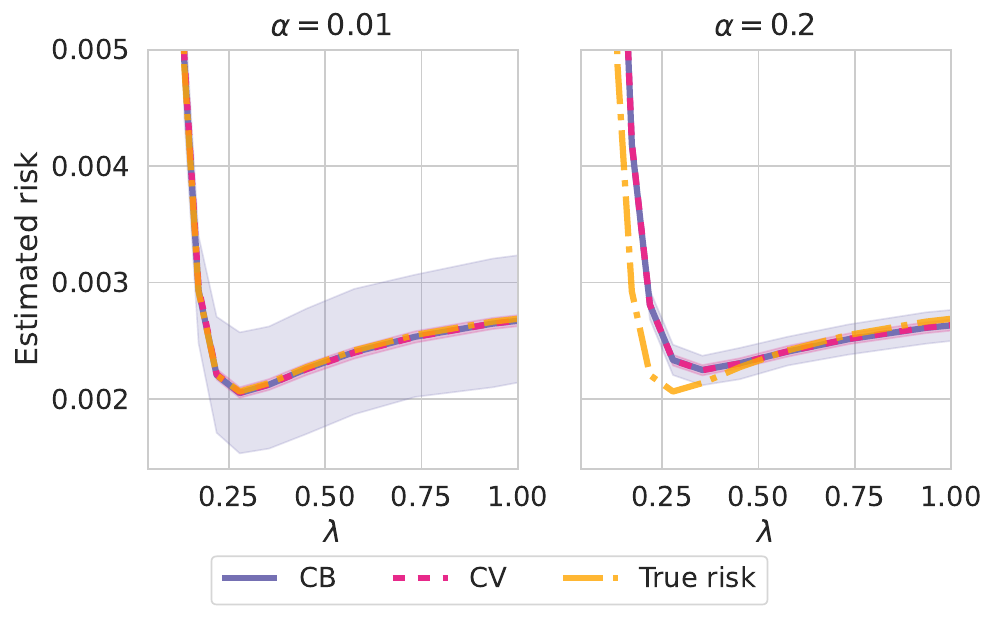}
    \caption{}
    \label{fig: fused lasso risk}
    \end{subfigure}
    \begin{subfigure}{.46\textwidth}
        \includegraphics[width=\textwidth]{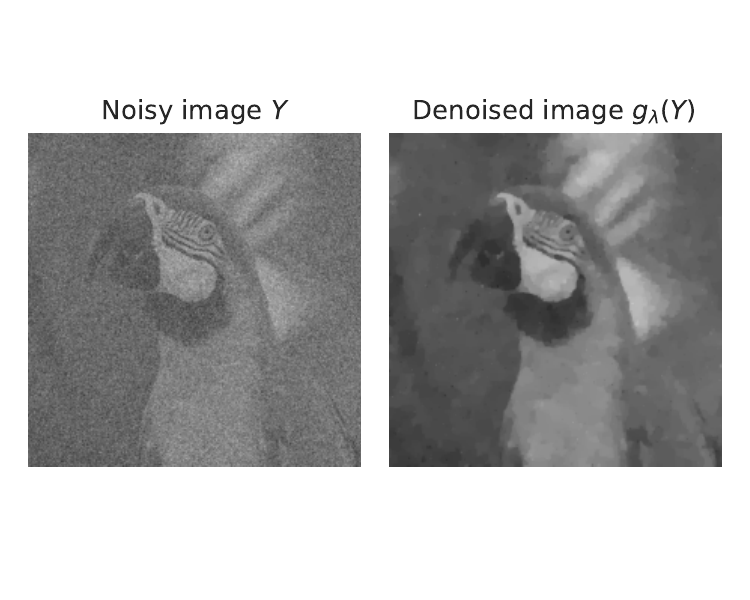}
    \caption{}
    \label{fig: parrot}
    \end{subfigure}
    \caption{(a) Risk of the fused lasso estimator as a function of the regularization parameter $\lambda$ with $\alpha=0.01$ (left) or $\alpha=0.2$ (right), estimated by antithetic CV (pink dashed line) and coupled bootstrap (purple line). The shaded regions represent error bars from 100 independent replicates. The true risk is shown in orange. (b) One noisy image $Y$ and the corresponding denoised image $g_\lambda(Y)$, with the optimal $\lambda$ picked by our method.}
    \label{fig: fused lasso}
\end{figure}
 
\section{Concluding remarks}
\label{sec:conclusion}

In this paper, we propose a new cross-validation scheme based on an equicorrelated Gaussian randomization scheme for estimating prediction error.
Rather than relying on i.i.d. assumptions, our method requires that the sufficient statistic underlying the estimated prediction function is asymptotically normal with an estimable covariance structure. 
Unlike standard CV, which partitions the data into sample folds, our approach partitions the information contained within the sufficient statistic to construct train--test data pairs.
Therefore, our method offers a straightforward, computationally-efficient alternative to standard CV, particularly when i.i.d. assumptions are violated or when forming representative folds, such as in categorical data analysis, is difficult.

Although related in spirit to recently introduced data fission and data thinning methods, which split information within each observation, our cross-validation method provides a principled approach to controlling the key bias--variance tradeoff by exploiting the antithetic correlation among the external randomization variables.
Moreover, by randomizing the asymptotically normal sufficient statistic, our approach naturally extends to a broad class of data types from exponential-family distributions, yielding a relatively distribution-free method that preserves the same bias--variance tradeoff as observed under exact normality.

Several promising directions for future research remain, as outlined below.
\begin{enumerate}[leftmargin=*]
\item[(1)] Our paper establishes that, with the antithetic correlation between our randomization variables, bias of our cross-validated estimator can be reduced without an increase in variance for a large class of weakly differentiable estimators. However, discontinuous estimators, such as hard-thresholded estimators or estimators from decision trees, do not belong to this class. We investigated the former example and observed that, although bias reduction with a small $ \alpha $ involves some increase in variance, this cost still remains much smaller than that of the CB method. Extending this analysis to estimators beyond the class of weakly differentiable functions remains a topic for  theoretical investigation.
\item[(2)] As with data fission and data thinning methods for normal data, our antithetic CV method relies on access to the covariance of the sufficient statistic, or on a reasonably accurate estimate of it. Our empirical analysis demonstrates that a plug-in estimator---as in the logistic regression example---using a suitable parametric or nonparametric method, such as the bootstrap, performs well in practice. However, it remains to be investigated how the bias--variance properties extend in theory when employing a consistent plug-in estimator of the covariance matrix.
\item[(3)] The antithetic randomization approach presented in our paper may hold promise for other tasks, due to the ideal bias--variance tradeoff it strikes. For example, it can be applied to estimate gradients of functions with unavailable analytical gradients, as is common in black-box, gradient-free optimization \citep{spall1992multivariate}. More specifically, a proxy for the function's gradient can be obtained by applying antithetic randomization to obtain the gradient of a convolution-smoothed version of the function, which can then be used to perform gradient descent in such optimization problems.
\end{enumerate}

\section*{Acknowledgments}
S.P. acknowledges support from NSF CAREER Award DMS-2337882.

\bibliographystyle{abbrvnat}  
\bibliography{paper.bbl}

\newpage

\appendix

\section{Proof of Proposition~\ref{prop: SURE}}
\label{prf: prop SURE}
\begin{proof}[Proof of Proposition~\ref{prop: SURE}]
    Note that
    \begin{align*}
        \tilde\cv_\alpha &= \|Y - g*\varphi_{\alpha\sigma^2}(Y)\|_2^2 + \frac{2}{\sqrt\alpha}\EE{ \omega\tran g(Y+\sqrt\alpha\omega)\mid Y },\\
        \mathrm{SURE}(g*\varphi_{\alpha\sigma^2}) &= \|Y - g*\varphi_{\alpha\sigma^2}(Y)\|_2^2 + 2\sigma^2\nabla\cdot (g*\varphi_{\alpha\sigma^2})(Y).
    \end{align*}
    It suffices to show that
    \begin{align*}
        \sigma^2 \nabla\cdot (g*\varphi_{\alpha\sigma^2})(Y) = \frac{1}{\sqrt\alpha}\EE{\omega\tran g(Y+\sqrt\alpha\omega)\mid Y }.
    \end{align*}
    This is true because
    \begin{align*}
        \nabla\cdot (g*\varphi_{\alpha\sigma^2})(y) &= \sum_{i=1}^n \nabla_{y_i} g_i*\varphi_{\alpha\sigma^2}(y)=\sum_{i=1}^n \nabla_{y_i} \int_{\R} g_i(x) \varphi_{\alpha\sigma^2}(y-x)\rd x\\
        &=\sum_{i=1}^n \int_{\R} g_i(x)(-\frac{y_i-x }{\alpha\sigma^2}) \varphi_{\alpha\sigma^2}(y-x)\rd x\\
        &=-\frac{1}{\alpha\sigma^2} \int g(x)\tran (y-x) \varphi_{\alpha\sigma^2}(y-x)\rd x\\
        &=\frac{1}{\alpha\sigma^2} \int g(y+\ep)\tran \ep \varphi_{\alpha\sigma^2}(\ep)\rd \ep\\
        &=\frac{1}{\alpha\sigma^2} \EE[\ep\sim\N(0,\alpha\sigma^2 I)]{g(y+\ep)\tran \ep  }\\
        &=\frac{1}{\sqrt{\alpha}\sigma^2} \EE[\omega\sim\N(0,\sigma^2I)] {g(y+\sqrt\alpha\omega)\tran\omega }.
    \end{align*}
\end{proof}

\section{Proofs for Section~\ref{sec: theory}}

\subsection{Proof of Theorem~\ref{thm: reducible variance}}
\label{prf: thm reducible variance}
\begin{proof}[Proof of Theorem~\ref{thm: reducible variance}]
We first write
\begin{align*}
    \cv_\alpha&=\frac1K\sum_{k=1}^K \|Y-\frac{1}{\sqrt\alpha}\omega^{(k)} - g(Y +\sqrt\alpha\omega^{(k)} )\|_2^2 - \frac{1}{\alpha}\|\omega^{(k)}\|_2^2\\
    &= \underbrace{\frac1K\sum_{k=1}^K \left[ \|Y-g(Y+\sqrt\alpha\omega^{(k)})\|_2^2 \right]}_{(\Rom{1})} + 
    \underbrace{\frac1K\sum_{k=1}^K \frac{2}{\sqrt\alpha}\langle  \omega^{(k)} , g(Y+\sqrt\alpha\omega^{(k)})\rangle}_{(\Rom{2})}.
\end{align*}
By Lemma~\ref{lem: first term}, $\Var{(\Rom{1}) \mid y } $ converges in $L_1$ to 0. By Lemma~\ref{lem: second term}, $\Var{(\Rom{2})\mid Y } $ converges in $L_1$ to $\Var{\frac{2}{K}\sum_{k=1}^K (\omega^{(k)})\tran \nabla g(Y) \omega^{(k)} \mid Y }$. When $j\neq k$, $\Cov{\omega^{(j)}, \omega^{(k)} }=\rho \sigma^2 I$ where $\rho=-\frac{1}{K-1} $. So we have 
\begin{align*}
    &\Var{\frac{1}{K}\sum_k (\omega^{(k)})\tran \nabla g(Y) \omega^{(k)} \mid Y }\\
    &\qquad =\frac{1}{K^2}\left(K\cdot \Var{\omega\tran \nabla g(Y)\omega } + K(K-1) \Cov{(\omega^{(1)})\tran \nabla g(Y) \omega^{(1)}, (\omega^{(2)})\tran \nabla g(Y) \omega^{(2)} }  \right).
\end{align*}
By Lemma~\ref{lem: gaussian quadratic covariance}, 
\begin{align*}
    &\Var{\omega\tran \nabla g(Y)\omega }=\sigma^4 (\|\nabla g(Y) \|_F^2 + \tr(\nabla g(Y)^2 ) ),\\
    &\Cov{(\omega^{(1)})\tran \nabla g(Y) \omega^{(1)}, (\omega^{(2)})\tran \nabla g(Y) \omega^{(2)} } =\frac{1}{(K-1)^2} \Var{\omega\tran \nabla g(Y)\omega }.
\end{align*}
Therefore,
\begin{align*}
    \Var{\frac{1}{K}\sum_k (\omega^{(k)})\tran \nabla g(Y) \omega^{(k)} \mid Y } &=\frac{1}{K^2}\left(K + K(K-1) \frac{1}{(K-1)^2}  \right) \Var{\omega\tran \nabla g(Y)\omega } \\
    &=\frac{\sigma^4}{K-1}(\|\nabla g(Y) \|_F^2 + \tr(\nabla g(Y)^2 ) ).
\end{align*}
This completes the proof.
\end{proof}

\begin{lemma}[First term $(\Rom{1})$]\label{lem: first term}
    Assume that $\EE{\|g(Y)\|_2^4}<\infty$. Then as $\alpha\to0$,
    \begin{align*}
    \Var{ \frac1K\sum_{k=1}^K \|Y -g(Y + \sqrt\alpha\omega^{(k)}) \|_2^2 \mid Y }\stackrel{L_1}{\to} 0 .
    \end{align*}
\end{lemma}
\begin{proof}
    Because $\Var{X_1+\ldots+X_K}\leq K\cdot(\Var{X_1}+\ldots+\Var{X_K})$,
    it suffices to show that as $\alpha\to0$,
    \begin{align*}
        \Var{ \|Y - g(Y + \sqrt\alpha\omega)\tran \|_2^2 \mid Y } \stackrel{L_1}{\to} 0,
    \end{align*}
    where $\omega\sim \N(0, \sigma^2 I_n)$.
    It suffices to show that $\EE{\|Y - g(Y+\sqrt\alpha\omega) \|_2^4\mid Y }\stackrel{L_1}{\to}\|Y-g(Y)\|_2^4$, which is equivalent to showing that $\EE{\|g(Y+\sqrt\alpha\omega)\|_2^4 \mid Y }\stackrel{L_1}{\to}\|g(Y)\|_2^4$. This is true by applying Lemma~\ref{lem: approximation to identity} with $f(y)=\|g(y)\|_2^4$.
\end{proof}

\begin{lemma}[Second term $(\Rom{2})$]\label{lem: second term}
    Assume that $\EE{\|\nabla g(Y)\|_F^2}<\infty$. Then as $\alpha\to0$,
    \begin{align*}
    \Var{\frac2K\sum_{k=1}^K \langle \frac{1}{\sqrt\alpha}\omega^{(k)}, g(Y+\sqrt\alpha\omega^{(k)})\rangle \mid Y }\stackrel{L_1}{\to}\Var{\frac{2}{K}\sum_{k=1}^K (\omega^{(k)})\tran \nabla g(Y) \omega^{(k)} \mid Y}.
    \end{align*}
\end{lemma}
\begin{proof}
Since all the components of $g$ are almost differentiable, we have
\begin{align*}
\langle\frac{1}{\sqrt\alpha} \omega, g(Y+\sqrt\alpha\omega)\rangle&=\frac{1}{\sqrt\alpha}\sum_{i=1}^n \omega_i g_i(Y + \sqrt\alpha\omega)\\
&=\frac{1}{\sqrt\alpha}\sum_{i=1}^n \omega_i \left[g_i(Y) + \int_0^1 \sqrt\alpha\omega \cdot \nabla g_i(Y+t\sqrt\alpha\omega )\rd t \right]  \\
&=\frac{1}{\sqrt\alpha}\omega\tran g(Y) + \sum_{i=1}^n \omega_i \int_0^1 \omega\cdot \nabla g_i(Y+t\sqrt\alpha\omega )\rd t\\
&=\frac{1}{\sqrt\alpha}\omega\tran g(Y) + \omega\tran \int_0^1  \nabla g(Y+t\sqrt\alpha\omega )\rd t\, \omega,
\end{align*}
where $\nabla g$ is the $n\times n$ matrix with the $i$-th row being $\nabla g_i$. 
Averaging over $k=1,\ldots,K$ and using the fact that $\sum_{k=1}^K\omega^{(k)}=0$, we have
\begin{align*}
    \frac1K\sum_{k=1}^K \langle \frac{1}{\sqrt\alpha}\omega^{(k)}, g(Y+\alpha\omega^{(k)})\rangle = 
     \frac{1}{K}\sum_{k=1}^K  (\omega^{(k)})\tran \int_0^1 \nabla g(Y+t\sqrt\alpha\omega^{(k)})\rd t\, \omega^{(k)}.
\end{align*}
Denote each summand as $\Gamma_\alpha(\omega^{(k)}, Y)$, where $\Gamma_\alpha(\omega,Y)=\omega\tran\int_0^1  \nabla g(Y+t\sqrt\alpha\omega )\rd t \omega $ and $\Gamma_0(\omega, Y)=\omega\tran \nabla g(Y) \omega$.
The claim of the lemma is equivalent to
\begin{align*}
    \Var{\frac1K\sum_{k=1}^K \Gamma_\alpha(\omega^{(k)}, Y )\mid Y }\stackrel{L_1}{\to}\Var{\frac1K\sum_{k=1}^K \Gamma_0(\omega^{(k)}, Y ) \mid Y },\text{ as }\alpha\to0.
\end{align*}
It suffices to show that as $\alpha\to0$,
\begin{align*}
    \Var{\Gamma_\alpha(\omega,Y)\mid Y} &\stackrel{L_1}{\to} \Var{\Gamma_0(\omega,Y)\mid Y},\\
    \Cov{\Gamma_\alpha(\omega^{(1)}, Y), \Gamma_\alpha(\omega^{(2)}, Y)\mid Y} &\stackrel{L_1}{\to} \Cov{\Gamma_0(\omega^{(1)}, Y), \Gamma_0(\omega^{(2)}, Y)\mid Y}.
\end{align*}
The conditional variance and conditional covariance involve the conditional expectation of the following terms:
\begin{align*}
    &\omega\tran \int_0^1 \nabla g(Y+t\sqrt\alpha\omega)\rd t\, \omega = \sum_{i,j=1}^n \int_0^1 \nabla_{j} g_i(Y+t\sqrt\alpha\omega) \omega_i\omega_j \rd t\\
    &(\omega\tran\int_0^1 \nabla g(Y+t\sqrt\alpha\omega)\rd t\omega)^2 =\\
    &\qquad\qquad \sum_{i,j,k,l}\int_{[0,1]^2} \nabla_j g_i (Y+t_1\sqrt\alpha\omega) \nabla_l g_k(Y+t_2\sqrt\alpha\omega) \omega_i\omega_j\omega_k\omega_l\rd t_1 \rd t_2\\
    &({\omega^{(1)}}\tran \int_0^1 \nabla g(Y+t\sqrt\alpha\omega^{(1)})\rd t \omega^{(1)})({\omega^{(2)}}\tran \int_0^1 \nabla g(Y+t\sqrt\alpha\omega^{(2)})\rd t \omega^{(2)})=\\
    &\qquad\qquad\sum_{i,j,k,l}\int_{[0,1]^2} \nabla_j g_i(Y+t_1\sqrt\alpha\omega^{(1)}) \nabla_l g_k(Y+t_2\sqrt\alpha\omega^{(2)}) \omega_i^{(1)}\omega_j^{(1)}\omega_k^{(2)}\omega_l^{(2)}\rd t_1 \rd t_2
\end{align*}
Consider the quantities
\begin{align*}
    &\omega_i\omega_j \int_0^1 \nabla_j g_i (Y+t\sqrt\alpha\omega)\rd t\\
    & \omega_i\omega_j\omega_k\omega_l \int_0^1\int_0^1 \nabla_j g_i (Y+t_1\sqrt\alpha\omega) \nabla_l g_k(Y+t_2\sqrt\alpha\omega)\rd t_1\rd t_2\\
    &\omega_i^{(1)}\omega_j^{(1)}\omega_k^{(2)}\omega_l^{(2)} \int_0^1\int_0^1 \nabla_j g_i(Y+t_1\sqrt\alpha\omega^{(1)}) \nabla_l g_k(Y+t_2\sqrt\alpha\omega^{(2)})\rd t_1\rd t_2 .
\end{align*}
We need to show that their conditional expectation $\EE{\cdot\mid Y}$ converges in $L_1$ to their corresponding expected value at $\alpha=0$.
For the first term, we apply Lemma~\ref{lem: L1 t} with $f=\nabla_j g_i $ and $h(\omega)=\omega_i\omega_j$.
For the second term, we apply Lemma~\ref{lem: L1 t2} with $f_1=\nabla_j g_i$, $f_2= \nabla_l g_k $, and $h(\omega)=\omega_i\omega_j\omega_k\omega_l$. 
For the third term, we can apply Lemma~\ref{lem: L1 t2} as well. 
\end{proof}

\subsection{Proof of Theorem~\ref{thm: irreducible variance}}
\label{prf: irreducible}

\begin{proof}[Proof of Theorem \ref{thm: irreducible variance}]
    Note that 
    \begin{align*}
        \EE{\cv_\alpha\mid Y} &=\|Y\|_2^2 +\EE{\|g(Y+\sqrt\alpha\omega)\|_2^2 \mid Y } -2Y\tran\EE{g(Y+\sqrt\alpha\omega)\mid Y} \\
        &+ \frac{2}{\sqrt\alpha}\EE{\omega\tran g(Y+\sqrt\alpha\omega) \mid Y }.
    \end{align*}
    It suffices to show that as $\alpha\to0$,
    \begin{align*}
        \EE{\|g(Y+\sqrt\alpha\omega)\|_2^2 \mid Y} &\stackrel{L_2}{\to} \|g(Y)\|_2^2,\\
        \EE{g(Y+\sqrt\alpha\omega)\tran Y \mid Y} &\stackrel{L_2}{\to} g(Y)\tran Y,\\
        \EE{\frac{1}{\sqrt\alpha}\omega\tran g(Y+\sqrt\alpha\omega)\mid Y } &\stackrel{L_2}{\to} \sigma^2 \tr(\nabla g(Y)).
    \end{align*}
    Applying Lemma~\ref{lem: L1} with $f(y)=\|g(y)\|_2^4$ and $h(\omega)=1$ proves the first line. 
    Applying Lemma~\ref{lem: L1} with $f(y)=g_i(y)g_j(y)$ and $h(\omega)=1$ ($1\leq i,j\leq n$) proves the first two lines.
    For the third limit, recall that 
    \begin{align*}
        \frac{1}{\sqrt\alpha}\omega\tran g(Y+\sqrt\alpha\omega)=\frac{1}{\sqrt\alpha}\omega\tran g(Y) + \omega\tran \int_0^1  \nabla g(Y+t\sqrt\alpha\omega )\rd t \omega.
    \end{align*}
    Because $\EE{\omega\tran g(Y)\mid Y }=0$, we have
    \begin{align*}
        \EE{\frac{1}{\sqrt\alpha}\omega\tran g(Y+\sqrt\alpha\omega) \mid Y} &= \EE{\omega\tran \int_0^1 \nabla g(Y+t\sqrt\alpha\omega )\rd t \omega \mid Y }\\
        &=\sum_{i,j=1}^n \EE{\int_0^1 \nabla_j g_i(Y+t\sqrt\alpha\omega )\rd t  \omega_i\omega_j \mid Y}.
    \end{align*}
    Applying Lemma~\ref{lem: L1 t} with $f(y)=\nabla_j g_i(y)\nabla_l g_k(y) $  and $h(\omega)=\omega_i\omega_j\omega_k\omega_l $ ($1\leq i,j,k,l\leq n$) proves the third line.
\end{proof}

\section{Proofs for Section~\ref{sec:glm}}

\subsection{Proof of Theorem~\ref{thm: glm bias}}
\label{prf: glm bias}
\begin{proof}[Proof of Theorem~\ref{thm: glm bias}]
    The bias can be decomposed as
    \begin{align*}
        &|\EE{\cv_{n,\alpha}} - \PE_n(g)|\\
        &\qquad= \Big |\EE{\frakA_n(T_n+\sqrt\alpha\omega)-\frakg_n(T_n+\sqrt\alpha\omega)\tran (T_n-\frac{1}{\sqrt\alpha}\omega) -\frakA_n(T_n) + \frakg_n(T_n)\tran m_n }\Big|\\
        &\qquad \leq \Big|\EE{\frakA_n(T_n+\sqrt\alpha\omega)} - \EE{\frakA_n(T_n)}\Big| + \Big|\EE{\frakg_n(T_n+\sqrt\alpha\omega)\tran T_n - \frakg_n(T_n)\tran T_n } \Big| \\
        &\qquad\qquad + \Big| \EE{\frac{1}{\sqrt\alpha}\omega\tran \frakg_n(T_n+\sqrt\alpha\omega) -\frakg_n(T_n)\tran (T_n - m_n) }\Big|.\numberthis\label{equ: glm bias prf 1}
    \end{align*}
    Because the density $q_n$ of $T_n$ satisfies Assumption~\ref{assump: log density q_n}, and $\frakg_n$, $\frakA_n$ are integrable w.r.t. $q_n$, we can apply Lemma~\ref{lem: L1} to show that, for $n\geq N_0$,
    \begin{align*}
        &\lim_{\alpha\to0}\EE{\frakA_n(T_n + \sqrt\alpha\omega) } = \EE{\frakA_n(T_n)},\\
        &\lim_{\alpha\to0}\EE{\frakg_n(T_n + \sqrt\alpha\omega)\tran T_n } = \EE{\frakg_n(T_n)\tran T_n}.
    \end{align*}
    Thus, the first two terms in the upper bound~\eqref{equ: glm bias prf 1} vanish as $\alpha\to0$.
    It remains to show that
    \begin{align}\label{equ: glm bias prf 2}
        \lim_{n\to\infty}\lim_{\alpha\to0}\EE{\frac{1}{\sqrt\alpha}\omega\tran \frakg_n(T_n+\sqrt\alpha\omega) -\frakg_n(T_n)\tran (T_n - m_n) }=0.
    \end{align}
    Because $\frakg_n$ is weakly differentiable, we can write
    \begin{align*}
        \EE{\frac{1}{\sqrt\alpha}\omega\tran \frakg_n(T_n+\sqrt\alpha\omega)}=\EE{\frac{1}{\sqrt\alpha}\omega\tran \frakg_n(T_n)} + \EE{\omega\tran \int_0^1 \nabla \frakg_n(T_n+t\sqrt\alpha\omega) \rd t \omega}.
    \end{align*}
    The first expectation is 0. For the second term, Lemma~\ref{lem: L1 t} shows that
    \begin{align*}
        \lim_{\alpha\to0}\EE{\omega\tran \int_0^1 \nabla \frakg_n(T_n+t\sqrt\alpha\omega) \rd t \omega} = \EE{\omega\tran \nabla \frakg_n(T_n) \omega} =\EE{\nabla\cdot\frakg_n(T_n)}.
    \end{align*}
    By Assumption~\ref{assump: stein discrepancy}, we have
    \begin{align*}
        \lim_{n\to\infty}\EE{\nabla\cdot \frakg_n(T_n)} - \EE{\frakg_n(T_n)\tran (T_n-m_n) }=0.
    \end{align*}
    This proves Equation~\eqref{equ: glm bias prf 2}.
\end{proof}

\subsection{Proof of Theorem~\ref{thm: glm var}}
\label{prf: glm var}
\begin{proof}[Proof of Theorem~\ref{thm: glm var}]
    We write 
    \begin{align*}
    \cv_{n,\alpha}=(\Rom{1}) + (\Rom{2}),
    \end{align*} 
    where
    \begin{align*}        
        (\Rom{1})&=\frac1{K}\sum_{k=1}^K\Big\{\mathfrak A_n( T_n+\sqrt\alpha\omega^{(k)}) - \mathfrak g_n(T_n + \sqrt\alpha\omega^{(k)} )\tran T_n   \Big\} - n^{-1/2} \log h_n(Y_n)\\
        (\Rom{2})&=\frac1K\sum_{k=1}^K \frac{1}{\sqrt\alpha}\frakg_n(T_n+\sqrt\alpha\omega^{(k)})\tran \omega^{(k)}
    \end{align*}    
    Because $\frakA_n^2,\|\frakg_n\|_2^4$ are integrable under $q_n$, and $q_n$ satisfies Assumption~\ref{assump: log density q_n}, we can apply Lemma~\ref{lem: L1} with $f(y)=\frakA_n^2(y)$ and $f(y)=\|\frakg_n\|_2^4 $ to show that, when $n\geq N_0$,
    \begin{align*}
        \lim_{\alpha\to0}\EE{\Var{(\Rom{1}) \mid Y_n} }=0.
    \end{align*}
    For the second term, because $\frakg_n$ is weakly differentiable, we can write
    \begin{align*}
        (\Rom{2})&=\frac1K\sum_{k=1}^K \frac{1}{\sqrt\alpha}(\omega^{(k)})\tran \frakg_n(T_n) +\frac1K\sum_{k=1}^K  \int_0^1 \omega^{(k)}\cdot \nabla\frakg_n(T_n+t\sqrt\alpha\omega^{(k)}) \rd t \omega^{(k)} \\
        &=\frac1K\sum_{k=1}^K  \int_0^1 \omega^{(k)}\cdot \nabla\frakg_n(T_n+t\sqrt\alpha\omega^{(k)}) \rd t \omega^{(k)},
    \end{align*}
    where the second equality is due to the zero-sum constraint $\sum_{k=1}^K\omega^{(k)}=0$. The rest of the proof is similar to the proof of Theorem~\ref{thm: reducible variance}.
\end{proof}

\section{Technical lemmas}
\begin{lemma}\label{lem: log p condition}
    Let $p$ be a continuous density on $\R^n$ with respect to the Lebesgue measure. Suppose there exist a constant $L>0$ such that for all $x,x'\in\R^n$, 
    \begin{align*}
        |\log p(x)-\log p(x')|\leq  L \|x-x'\|_2^2.
    \end{align*}
    Let $\varphi$ be the density of the standard normal distribution $\N(0,I_n)$ with respect to the Lebesgue measure. 
    Then the following results hold:
    \begin{enumerate}
        \item Let $f:\R^n\to\R$ be $L_1$-integrable with respect to $p$. 
        Let $h:\R^n\to\R$ be be integrable w.r.t. $\N(0,(1+\delta_0)I_n)$ for some $\delta_0>0$.
        Then there exists $\ep_0>0$, such that for all $\ep\in[0,\ep_0]$, 
        \begin{align*}
            |h(\omega)| \varphi(\omega)\int_{\R^n} |f(x+\ep\omega)|p(x)\rd x \leq q(\omega),
        \end{align*}
        where $q$ is an $L_1$-integrable function on $\R^n$.

        \item Let $f_1,f_2$ be $L_2$-integrable functions with respect to $p$. Let $h:\R^n\to\R$ be integrable w.r.t. $\N(0,(1+\delta_0)I_n)$ for some $\delta_0>0$. Then there exists $\ep_0>0$, such that when $\ep_1,\ep_2\in[0,\ep_0]$,
        \begin{align*}
            |h(\omega)| \varphi(\omega)\int |f_1(x+\ep_1\omega)f_2(x+\ep_2\omega)| p(x)\rd x \leq q(\omega),
        \end{align*}
        where $q$ is an $L_1$-integrable function on $\R^n$.
    \end{enumerate}
    In particular, if $p$ is a Gaussian density, then $\log p$ is a quadratic function, so the condition on $\log p$ in the lemma is satisfied. 
\end{lemma}

\begin{proof}
    By the assumption on $\log p$, we have
    \begin{align*}
        |\log p(x-\ep\omega) - \log p(x)|\leq L \ep^2\|\omega\|_2^2.
    \end{align*}
    Then
    \begin{align*}
        \frac{p(x-\ep\omega)}{p(x)}=\exp(\log p(x-\ep\omega) - \log p(x) )\leq \exp(L\ep^2 \|\omega\|_2^2 ).
    \end{align*}
    Thus, we obtain
    \begin{align*}
        \varphi(\omega)\frac{p(x-\ep\omega)}{p(x)}\leq \frac{1}{(2\pi)^{n/2}}\exp(-(\frac12 - L\ep^2 )\|\omega\|_2^2  ).
    \end{align*}
    Take $\ep_0= (\frac{\delta_0}{4L(1+\delta_0)} )^{1/2}$. When $\ep\leq \ep_0$, we have $\frac12- L\ep^2\geq \frac{1}{2}-L\ep_0^2=\frac{1+\delta_0/2}{2(1+\delta_0)} > \frac{1}{2(1+\delta_0)}$. Define $q(\omega):= |h(\omega)| \frac{1}{(2\pi)^{n/2}} \exp(-\frac{1}{2(1+\delta_0)}\|\omega\|_2^2 )$, which is an integrable function by the assumption.
    We obtain
    \begin{align*}
        |h(\omega)| \varphi(\omega) p(x-\ep\omega)&\leq |h(\omega)|\frac{1}{(2\pi)^{n/2}}\exp(-(\frac12 - L\ep^2 )\|\omega\|_2^2  )\cdot p(x)\\
        &\leq |h(\omega)|\frac{1}{(2\pi)^{n/2}}\exp(-\frac{1}{2(1+\delta_0)} \|\omega\|_2^2  ) \cdot p(x)\\
        &=q(\omega)p(x).
    \end{align*}
    Therefore,
    \begin{align*}
        |h(\omega) |\varphi(\omega)\int |f(x+\ep\omega)| p(x)\rd x &=|h(\omega) |\varphi(\omega)\int |f(x)| p(x-\ep\omega)\rd x\\
        &=\int |f(x)| |h(\omega)| \varphi(\omega) p(x-\ep\omega)\rd x\\
        &\leq \int |f(x)| q(\omega) p(x)\rd x\\
        &=q(\omega)\cdot \int |f(x)| p(x)\rd x,
    \end{align*}
    which is integrable because $q$ is integrable and $\int |f(x)| p(x)\rd x<\infty$. This proves the first claim.

    Next, suppose $f_1^2,f_2^2$ are integrable functions with respect to $p$. From the first claim, there exists $\ep_0$, such that for all $\ep_1,\ep_2\in[0,\ep_0]$, 
    \begin{align*}
        |h(\omega)| \varphi(\omega)\int f_1(x+\ep_1\omega)^2 p(x)\rd x\leq q_1(\omega),\\
        |h(\omega)|\varphi(\omega)\int f_2(x+\ep_2\omega)^2 p(x)\rd x\leq q_2(\omega),
    \end{align*}
    where $q_1,q_2$ are integrable functions.
    By the Cauchy-Schwarz inequality, we have
    \begin{align*}
        &|h(\omega)| \varphi(\omega)\int |f_1(x+\ep_1\omega)f_2(x+\ep_2\omega) |p(x)\rd x\\ 
        &\qquad\leq \left(|h(\omega)| \varphi(\omega) \int f_1(x+\ep_1\omega)^2 p(x)\rd x \right)^{1/2} \cdot \left(|h(\omega)| \varphi(\omega) \int f_2(x+\ep_2\omega)^2 p(x)\rd x \right)^{1/2}\\
        &\qquad \leq q_1(\omega)^{1/2}q_2(\omega)^{1/2},
    \end{align*}
    which is $L_1$ integrable, proving the second claim.
\end{proof}

\begin{lemma}
    \label{lem: L1}
    Let $p$ be a density on $\R^n$ satisfying the condition in Lemma~\ref{lem: log p condition}.
    Let $f:\R^n\to\R$ be a function that is integrable w.r.t. the density $p$.
    Let $h:\R^n\to\R$ be a function that is integrable w.r.t. $\N(0,(1+\delta_0)I_n)$ for some $\delta_0>0$.
    Then as $\alpha\downarrow 0$,
    \begin{align*}
        \EE{f(Y+\sqrt{\alpha}\omega ) h(\omega)  \mid Y}\stackrel{L_1}{\rightarrow} f(Y)\EE{h(\omega)},
    \end{align*}
    where the expectation is taken over $\omega\sim\N(0,I_n)$, and the $L_1$ convergence is with respect to $Y\sim p$.
\end{lemma}
\begin{proof}
    We start from the inequality
    \begin{align*}
        \Big |\int (f(y+\sqrt\alpha\omega)h(\omega) - f(y)h(\omega) ) \varphi(\omega)\rd\omega\Big|
        &\leq \int |f(y+\sqrt\alpha\omega) - f(y) |\cdot |h(\omega)| \varphi(\omega) \rd\omega.
    \end{align*}
    Multiplying by $p(y)$, integrating over $y$, and applying Fubini's theorem, we have
    \begin{align*}
        &\int \Big|\int (f(y+\sqrt\alpha\omega)h(\omega) -f(y)h(\omega) )\varphi(\omega)\rd \omega  \Big| p(y)\rd y\\
        &\qquad \leq \int\int  |f(y+\sqrt\alpha\omega) - f(y) |\cdot |h(\omega)| \varphi(\omega)\rd\omega p(y)\rd y\\
        &\qquad =\int |h(\omega)| \varphi(\omega) \int |f(y+\sqrt\alpha\omega)-f(y) | p(y)\rd y\rd\omega .
    \end{align*}
    By Lemma~\ref{lem: log p condition}, there exists $\ep_0>0$ such that, for all $\ep\in[0,\ep_0]$,  
    \begin{align*}
        |h(\omega)| \varphi(\omega) \int |f(y + \ep \omega)| p(y)\rd y \leq \bar f(\omega),
    \end{align*}
    where $\bar f(\omega)$ is $L_1$ integrable.
    When $\sqrt\alpha\leq \ep_0$, we have
    \begin{align*}
    |h(\omega)|\varphi(\omega) \int |f(y+\sqrt\alpha\omega) - f(y)| p(y) \rd y&\leq 2 \bar f(\omega).
    \end{align*}
    Thus, $|h(\omega)|\varphi(\omega) \int |f(y+\sqrt\alpha\omega) - f(y)| p(y) \rd y$ is dominated by an integrable function.
    Moreover, $\lim_{\alpha\to 0} \int |f(y+\sqrt\alpha\omega)-f(y)| p(y)\rd y=0$ (Lebesgue differentiation theorem or mean continuity theorem; proof by the fact that compactly supported continuous functions are dense in $L_1$).
    Applying the dominated convergence theorem, we have
    \begin{align*}
        &\lim_{\alpha\to0}\int |h(\omega)| \varphi(\omega) \int |f(y+\sqrt\alpha\omega) -f(y)| p(y)\rd y\rd \omega \\
        &\quad= \int|h(\omega)|  \varphi(\omega)  \lim_{\alpha\to0}\int |f(y+\sqrt\alpha\omega) -f(y)| p(y)\rd y\rd \omega\\
        &\quad =0.
    \end{align*}
\end{proof}

\begin{lemma}
    \label{lem: L1 t}
    Under the same condition as in Lemma~\ref{lem: L1}, as $\alpha\downarrow 0$, we have
    \begin{align*}
        \EE{\int_0^1 f(Y+t\sqrt{\alpha}\omega ) \rd t\cdot h(\omega) \mid Y}\stackrel{L_1}{\rightarrow} f(Y) \EE{h(\omega)},
    \end{align*}
    where the expectation is taken over $\omega\sim\N(0,I_n)$, and the $L_1$ convergence is with respect to $Y\sim p$.
\end{lemma}
\begin{proof}
    We start from
    \begin{align*}
        &\Big|\int_0^1 \int (f(y+t\sqrt\alpha\omega) h(\omega)-f(y)h(\omega) )\varphi(\omega)\rd \omega \rd t  \Big| \\
        &\qquad\leq \int_0^1\int |f(y+t\sqrt\alpha\omega)-f(y) |\cdot|h(\omega)|
        \varphi(\omega)\rd\omega\rd t.    
    \end{align*}
    Integrating over $y$ and applying Fubini's theorem, we have
    \begin{align*}
        &\int \Big|\int_0^1 \int  (f(y+t\sqrt\alpha\omega)h(\omega) -f(y)h(\omega) ) \varphi(\omega) \rd\omega \rd t \Big| p(y) \rd y  \\
        &\qquad\leq \int_0^1 \int |h(\omega)|  \varphi(\omega) \int |f(y+t\sqrt\alpha\omega) -f(y)| p(y)\rd y\rd \omega\rd t .
    \end{align*}
    When $\sqrt\alpha<\ep_0$, $t\sqrt\alpha<\ep_0$ for all $t\in[0,1]$.
    Following a similar argument as in the proof of Lemma~\ref{lem: L1}, we have $|h(\omega)|\varphi(\omega) \int |f(y+t\sqrt\alpha\omega) - f(y)| p(y) \rd y\leq 2 \bar f(\omega)$, which is integrable. Moreover, $\lim_{\alpha\to 0} \int |f(y+t\sqrt\alpha\omega)-f(y)|p(y)\rd y=0$ (mean continuity theorem).
    By dominated convergence theorem, the last display converges to 0 as $\alpha\to0$.
\end{proof}

\begin{lemma}
    \label{lem: L1 t2}
    Suppose $p$ is a density satisfying the condition in Lemma~\ref{lem: log p condition}.
    Suppose $f_1(y)^2,f_2(y)^2$ are integrable functions w.r.t. $p$.
    Suppose $h$ is integrable w.r.t. $\N(0, (1+\delta_0)I_n)$ for some $\delta_0>0$.
    Then as $\alpha\downarrow0$, we have
    \begin{align*}
        \EE{\int_0^1\int_0^1 f_1(Y+t_1\sqrt{\alpha}\omega ) f_2(Y+t_2\sqrt{\alpha}\omega )\rd t_1\rd t_2 \cdot h(\omega) \mid Y}\stackrel{L_1}{\rightarrow} f_1(Y)f_2(Y)\EE{h(\omega)} ,
    \end{align*}
    where the expectation is taken over $\omega\sim\N(0,I_n)$, and the $L_1$ convergence is with respect to $Y\sim p$.
\end{lemma}
\begin{proof}
    Because $f_1^2,f_2^2$ are integrable w.r.t. $p$, applying the second claim in Lemma~\ref{lem: log p condition}, there exists $\ep_0$ such that, when $\ep_1,\ep_2\in[0,\ep_0]$,
    \begin{align*}
        |h(\omega)|\varphi(\omega) \int |f_1(y+\ep_1\omega) f_2(y+\ep_2\omega)|  p(y)\rd y\leq \bar f(\omega),
    \end{align*}
    where $\bar f(\omega)$ is integrable.
    Following the same argument as in the proof of Lemma~\ref{lem: L1 t}, we have
    \begin{align*}
        &\int \bigg|\int_0^1\int_0^1 \int  (f_1(y+t_1\sqrt\alpha\omega)f_2(y+t_2\sqrt\alpha\omega)h(\omega) -f_1(y)f_2(y)h(\omega) )  \varphi(\omega) \rd\omega \rd t_1\rd t_2 \bigg| p(y) \rd y  \\
        &\quad\leq \int_0^1\int_0^1 \int |h(\omega)| \varphi(\omega) \int \big|f_1(y+t_1\sqrt\alpha\omega)f_2(y+t_2\sqrt\alpha\omega)-f_1(y)f_2(y)\big| p(y)\rd y\rd \omega\rd t_1 \rd t_2 .
    \end{align*}
    Following a similar argument as in the proof of Lemma~\ref{lem: L1 t}, when $\sqrt\alpha <\ep_0$, the function
    \begin{align*}
        &|h(\omega)|\varphi(\omega) \int |f_1(y+t_1\sqrt\alpha\omega)f_2(y+t_2\sqrt\alpha\omega) - f_1(y)f_2(y)|p(y) \rd y
    \end{align*}
    is dominated by an integrable function.
    Since $\lim_{\alpha\to 0} \int |f_1(y+t_1\sqrt\alpha\omega)f_2(y+t_2\sqrt\alpha\omega)-f_1(y)f_2(y)|p(y)\rd y=0$ (mean continuity theorem),
    the proof is completed by applying the dominated convergence theorem.
\end{proof}

\begin{lemma}
    \label{lem: gaussian quadratic covariance}
    Suppose $x,y\sim\N(0,I_n)$ and $\Cov{x,y}=\rho I_n$. For a matrix $A$, we have
    \begin{align*}
        \Cov{x\tran Ax, y\tran Ay}=\rho^2\Var{x\tran Ax}=\rho^2(\|A\|_F^2 + \tr(A^2) ).        
    \end{align*}
\end{lemma}

\begin{proof}[Proof of Lemma~\ref{lem: gaussian quadratic covariance}]
Let $x,z\iid \N(0,I_n)$ and $y=\rho x + \tilde\rho z$ where $\tilde \rho=\sqrt{1-\rho^2}$. Then $y\sim\N(0,I_n)$ and $\Cov{x,y}=\rho I_n$. 
We have
\begin{align*}
    \EE{(x\tran Ax)(y\tran Ay)}&=\EE{(x\tran Ax)(\rho^2 x\tran Ax + \tilde\rho^2 z\tran Az ) }\\
    &=\rho^2 \EE{(x\tran Ax)^2} + \tilde\rho^2 \EE{(x\tran Ax)}\EE{(z\tran Az)}\\
    &=\rho^2 \EE{(x\tran Ax)^2} + \tilde\rho^2 \tr(A)^2.
\end{align*}
Note that
\begin{align*}
    \EE{(x\tran Ax)^2}&=\EE{\sum_{i,j,k,l}A_{ij}A_{kl}x_ix_jx_kx_l }\\
    &=\EE{\sum_{i} A_{ii}^2x_i^4 + \sum_{i\neq j}A_{ij}^2 x_i^2x_j^2 + \sum_{i\neq j} A_{ij}A_{ji}x_i^2x_j^2 + \sum_{i\neq j} A_{ii}A_{jj}x_i^2x_j^2} \\
    &=3\sum_{i}A_{ii}^2 + \sum_{i\neq j}A_{ij}^2 + \sum_{i\neq j}A_{ij}A_{ji} + \sum_{i\neq j}A_{ii}A_{jj}\\
    &=\sum_{i,j}A_{ij}^2 + \sum_{i,j}A_{ij}A_{ji} + \sum_{i,j}A_{ii}A_{jj}\\
    &=\tr(A\tran A) + \tr(AA) + \tr(A)^2.
\end{align*}
Moreover, $\EE{x\tran Ax}=\tr(A)$. So we have
\begin{align*}
    \Cov{x\tran Ax, y\tran Ay}&=\rho^2\EE{(x\tran Ax)^2} + \tilde\rho^2\tr(A)^2 - \tr(A)^2\\
    &=\rho^2(\tr(A\tran A) + \tr(AA) + \tr(A)^2) + \tilde\rho^2\tr(A)^2 - \tr(A)^2\\
    &=\rho^2(\tr(A\tran A) + \tr(AA))\\
    &=\rho^2\Var{x\tran Ax}.
\end{align*}
\end{proof}

\section{Reducible variance for a non-weakly differentiable estimator}
\label{sec: indicator}
We consider the indicator function $g(y)=\Indc{y\geq\delta}$ as the prediction function with some threshold $\delta\in\R$. We let $n=1$ and $y\sim\N(\mu,1)$.

\begin{theorem}[Reducible variance for indicator function]\label{thm: indicator var}
    As $\alpha\to0$ and $K\to\infty$, we have
    \begin{align*}
        \EE{\Var{\cv_\alpha \mid y} } = O(\frac{1}{K\sqrt\alpha}).
    \end{align*}
\end{theorem}
\begin{proof}[Proof of Theorem~\ref{thm: indicator var}]
The proposed antithetic CV estimator can be written as
\begin{align*}
    \cv_\alpha=\underbrace{\frac1K\sum_{k=1}^K (y - \Indc{y + \sqrt\alpha\omega^{(k)} \geq\delta } )^2}_{(\Rom{1})} + \underbrace{\frac{2}{K\sqrt\alpha}\sum_{k=1}^K \omega^{(k)} \Indc{y + \sqrt\alpha\omega^{(k)}\geq\delta }}_{(\Rom{2})}.
\end{align*}
Conditioned on $y$, the variance of the first term $(\Rom{1})$ is bounded and goes to 0 as $\alpha\to0$. So we mainly focus on the conditional variance of the second sum $(\Rom{2})$.
Denote $c=c(y)=\frac{\delta-y}{\sqrt\alpha}$. 
Let $\omega,\omega^{(1)},\omega^{(2)}\sim \N(0,1)$ with $\mathrm{Corr}(\omega^{(1)}, \omega^{(2)}) = \rho=-\frac{1}{K-1}$, and define
\begin{align*}
    V=\Var{\omega\Indc{\omega>c}},\quad C=\Cov{\omega^{(1)}\Indc{\omega^{(1)}>c}, \omega^{(2)} \Indc{\omega^{(2)}>c} }.
\end{align*}
Then the conditional variance of $(\Rom{2})$ can be written as
\begin{align}\label{equ: cond var}
     \Var{(\Rom{2})\mid y}=\frac{4}{K^2\alpha}\Big[K\cdot V + K(K-1)\cdot C \Big] = \frac{4}{K\alpha} (V - \frac{1}{\rho} C).
\end{align}

\begin{lemma}[Expression of $V$]\label{lem: var}
    The variance term $V$ has the expression
    \begin{align*}
        V=\Var{\omega \Indc{\omega>c}} = c\varphi(c) + \bar\Phi(c) - \varphi(c)^2.
    \end{align*}
\end{lemma}
\begin{proof}[Proof of Lemma~\ref{lem: var}]
    The proof follows from direction calculations:
    \begin{align*}
    \EE{\omega\Indc{\omega>c}} & = \varphi(c),\quad \EE{\omega^2\Indc{\omega>c}} = c\varphi(c) + \bar\Phi(c),\\
    \Var{\omega\Indc{\omega>c}} &= c\varphi(c) + \bar\Phi(c) - \varphi(c)^2.
    \end{align*}    
\end{proof}

\begin{lemma}[Expression of $C$]\label{lem: cross term}
    The covariance term $C$ has the expression
    \begin{align*}
        &C=\Cov{\omega^{(1)}\Indc{\omega^{(1)}>c}, \omega^{(2)}\Indc{\omega^{(2)}>c}} =\\ 
        &\qquad \rho \bar F(c, c; \rho) + (1-\rho^2) \varphi_\rho(c, c) + 2\rho c \varphi(c) \bar\Phi(\frac{c-\rho c}{\sqrt{1-\rho^2}} ) - \varphi(c)^2,
    \end{align*}
    where $\varphi_\rho(h,k)$ is the joint density of a standard bivariate normal distribution with correlation $\rho$ evaluated at $(h,k)$; and $\bar F(h,k,\rho)$ is the corresponding survival function, i.e., $\bar F(h,k;\rho)=\PP{x>h, y>k}$ for $x,y\sim\N(0,1)$ with $\Cov{x,y}=\rho$.
\end{lemma}
\begin{proof}[Proof of Lemma~\ref{lem: cross term}]
    It remains to calculate the expectation \sloppy{$\EE{xy\Indc{x>c, y>c}}$}, where $(x,y)$ follows a standard bivariate normal distribution with correlation $\rho$.
    We can write
    \begin{align*}
        &\EE{xy\Indc{x>c, y>c}} = \int_c^\infty y J(y) \rd y,\\
        &\text{where }J(y) = \int_c^\infty x \varphi_\rho(x, y) \rd x.
    \end{align*}
    Using $x \varphi_\rho = \rho y \varphi_\rho - (1 - \rho^2) \partial_{x} \varphi_\rho$, we have
    \begin{align*}
        J(y) = \int_{c}^\infty [\rho y\varphi_\rho - (1 - \rho^2)\partial_{x} \varphi_\rho] \rd x=\rho y \int_c^\infty \varphi_\rho(x, y)\rd x + (1-\rho^2) \varphi_\rho(c, y),
    \end{align*}
    and thus
    \begin{align*}
        \int_c^\infty y J(y)\rd y &= \underbrace{\rho \int_{x>c,y>c} y^2 \varphi_\rho(x,y) \rd x \rd y}_{A_1} +  \underbrace{(1-\rho^2) \int_c^\infty y \varphi_\rho(c, y)\rd y}_{A_2}.
    \end{align*}
    Denote 
    \begin{align*}
        H(y) &= \int_c^\infty \varphi_\rho(x, y)\rd x= \varphi(y) \bar\Phi(\frac{c-\rho y}{\sqrt{1-\rho^2}} ), \\
        H'(y) &= -y H(y) + \varphi(y) \frac{\rho}{\sqrt{1-\rho^2}} \varphi(\frac{c-\rho y}{\sqrt{1-\rho^2}} ) = -y H(y) + \rho \varphi_\rho(c, y).
    \end{align*}
    Then we have
    \begin{align*}
        A_1 &= \rho \int_{c}^\infty y^2 H(y) \rd y \\
        &=\rho \int_{c}^\infty y \cdot (-H'(y) + \rho\varphi_\rho(c, y) )\rd y\\
        &=\rho \int_c^\infty -y \rd H(y) + \rho^2 \int_c^\infty y \varphi_\rho(c, y)\rd y\\
        &=\rho c H(c) + \rho \int_c^\infty H(y)\rd y + \rho^2 \int_c^\infty y \varphi_\rho(c, y)\rd y.
    \end{align*}
    So
    \begin{align*}
        A_1+A_2 = \rho c H(c) + \rho \int_c^\infty H(y)\rd y + \int_c^\infty y \varphi_\rho(c, y)\rd y.
    \end{align*}
    Moreover, we have $\int_c^\infty H(y)\rd y =\PP{x>c, y>c}=\bar F(c, c;\rho)$, and
    \begin{align*}
        \int_c^\infty y\varphi_\rho(c, y)\rd y &= \int_c^\infty (\rho c\varphi_\rho(c, y) - (1 - \rho^2)\partial_y \varphi_\rho(c, y) ) \rd y\\
        &=\rho c H(c) + (1-\rho^2) \varphi_\rho(c, c).
    \end{align*}
    Therefore,
    \begin{align*}
        A_1+ A_2 &= 2\rho c H(c) + \rho \bar F(c, c,\rho) + (1-\rho^2)\varphi_\rho(c, c)\\
        &=2\rho c \varphi(c) \bar\Phi(\frac{c-\rho c}{\sqrt{1-\rho^2}} ) + \rho \bar F(c, c,\rho) + (1-\rho^2)\varphi_\rho(c, c).
    \end{align*}
    The proof is completed by noting that $\EE{\omega^{(1)}\Indc{\omega^{(1)} >\delta}}=\varphi(c)$.
\end{proof}

Substituting the expressions of $V$ (Lemma~\ref{lem: var}) and $C$ (Lemma~\ref{lem: cross term}) into Equation~\eqref{equ: cond var}, we have
\begin{align*}
    \Var{(\Rom{2})\mid y} &= \frac{4}{K\alpha} \Big(c\varphi(c) + \bar\Phi(c) - \varphi(c)^2 \\
    &\quad - \bar F(c, c;\rho) - \frac{1-\rho^2}{\rho} \varphi_\rho(c, c) - 2c\varphi(c)\bar\Phi(\frac{c-\rho c}{\sqrt{1-\rho^2}} ) + \frac{1}{\rho}\varphi(c)^2 \Big) \numberthis\label{equ: cond var 2}
\end{align*}
where $c=\frac{\delta-y}{\sqrt\alpha}$.
We now study the rate of the expectation of each term in the parenthesis as $\alpha\to0$.
\begin{enumerate}
    \item The first term $\EE{c\varphi(c)}=O(\alpha)$: Note that
    \begin{align*}
        \EE{c\varphi(c) } & =\EE{\frac{\delta-y}{\sqrt\alpha} \varphi(\frac{\delta-y}{\sqrt\alpha} ) }=\EE{\frac{\delta-\mu-z}{\sqrt\alpha}\varphi(\frac{\delta-\mu-z}{\sqrt\alpha} ) },
    \end{align*}
    where $z\sim \N(0,1)$. The last expectation is equal to (see e.g.~\cite[Equation 111, Page 396]{owen1980table})
    \begin{align*}
        \EE{c\varphi(c)} = \frac{(\delta-\mu)\alpha}{(1+\alpha)^{3/2}}\varphi(\frac{\delta-\mu}{\sqrt{1+\alpha}})=O(\alpha).
    \end{align*}

    \item The second term $\EE{\bar\Phi(c)} =\Phi(\mu-\delta ) + O(\alpha)$:
    \begin{align*}
        \EE{\bar\Phi(c)}&=\EE{\Phi(\frac{y-\delta}{\sqrt\alpha} ) } = \PP{z\leq \frac{y-\delta}{\sqrt\alpha} }=\PP{\N(-\mu, 1+\alpha) \leq -\delta}=\Phi(\frac{\mu-\delta}{\sqrt{1+\alpha}})\\
        &=\Phi(\mu-\delta ) + O(\alpha).
    \end{align*} 

    \item The third term $\EE{\varphi(c)^2} = O(\sqrt\alpha) $:
    \begin{align*}
        \EE{\varphi(c)^2}&=\EE{\varphi(\frac{\delta-y}{\sqrt\alpha} )^2 } =\frac{1}{2\pi} \EE{\exp(-\frac{(\delta-y)^2}{\alpha} ) } = \frac{1}{2\pi}\EE{\exp(-\frac{(\delta-\mu-z)^2}{\alpha} ) }\\
        &=\frac{1}{\sqrt{2\pi}}\sqrt{\frac{\alpha}{\alpha+2} } \varphi(\frac{\delta-\mu}{\sqrt{1+\alpha/2} } ) =O(\sqrt\alpha).
    \end{align*}

    \item The fourth term $\EE{\bar F(c, c;\rho)} = \Phi(\mu-\delta) + O(\alpha) $:
    We have
    \begin{align*}
        \EE{\bar F(c, c;\rho)} &= \PP{x_1>\frac{\delta-y}{\sqrt\alpha}, x_2>\frac{\delta-y}{\sqrt\alpha}  } = \PP{x_1,x_2\geq \frac{\delta-\mu-z}{\sqrt\alpha}  }.
    \end{align*}
    where $(x_1,x_2)$ follow a standard bivariate normal distribution with correlation $\rho$.
    Let $w_i=\frac{x_i+z/\sqrt\alpha}{\sqrt{1+1/\alpha}}$ for $i=1,2$, which have correlation $\rho'=\frac{\rho+1/\alpha}{1+1/\alpha} $. Then
    \begin{align*}
        \EE{\bar F(c, c;\rho) } &= \PP{w_1,w_2\geq \frac{(\delta-\mu)/\sqrt\alpha}{\sqrt{1+1/\alpha}} } = \bar F\Big(\frac{(\delta-\mu)/\sqrt\alpha}{\sqrt{1+1/\alpha}} , \frac{(\delta-\mu)/\sqrt\alpha}{\sqrt{1+1/\alpha}} , \rho' \Big)\\
        &=\bar F(\frac{\delta-\mu}{\sqrt{1+\alpha}}, \frac{\delta-\mu}{\sqrt{1+\alpha}}, \frac{1+\rho\alpha}{1+\alpha} )\\
        &=\bar\Phi(\delta-\mu) + O(\alpha).
    \end{align*}
    The last line is because $\bar F$ is continuously differentiable in all three arguments, $\frac{\delta-\mu}{\sqrt{1+\alpha}}=\delta-\mu + O(\alpha)$ and $\frac{1+\rho\alpha}{1+\alpha} = 1+O(\alpha)$, as well as $\bar F(x, x, 1)=\bar\Phi(x)$.

    \item The fifth term $\frac{1-\rho^2}{\rho} \EE{\varphi_\rho(c, c) }=O(\sqrt\alpha)$:
    We have
    \begin{align*}
        \EE{\varphi_\rho(c, c) } &= \EE{\varphi_\rho(\frac{\delta-\mu+z}{\sqrt\alpha}, \frac{\delta-\mu+z}{\sqrt\alpha} ) }\\
        &=\frac{1}{2\pi\sqrt{(1-\rho^2)}}\int\exp(-\frac{z^2}{2} - \frac{((\delta-\mu+z)/\sqrt\alpha)^2 }{1+\rho} )  \rd z\\
        &=\frac{1}{\sqrt{2\pi(1-\rho^2)}} \frac{\sqrt\alpha}{\sqrt{\alpha+2/(1+\rho)}} \exp(\frac{2(\delta-\mu)^2}{((1+\rho)\alpha+2 )^2} - \frac{(\delta-\mu)^2}{(1+\rho)\alpha} )
    \end{align*}
    As $\alpha\to0$, the above quantity is $O(\sqrt\alpha \cdot \exp(-\frac{(\delta-\mu)^2}{(1+\rho)\alpha} )  )$, which is smaller than $O(\sqrt\alpha)$.
    
    \item The sixth term $\EE{2c\varphi(c)\bar\Phi(\frac{c-\rho c}{\sqrt{1-\rho^2}} ) } = O(\sqrt\alpha)$:
    Denote $k=\sqrt{\frac{1-\rho}{1+\rho}}$. Since $c\sim\N(\frac{\delta-\mu}{\sqrt\alpha} , \frac{1}{\alpha})$, we have
    \begin{align*}
        \EE{c\varphi(c)\bar\Phi(kc)} &= \int x\varphi(x)\bar\Phi(kx) \frac{\sqrt{\alpha}}{\sqrt{2\pi}} \exp(-\frac{\alpha}{2} (x - \frac{\delta-\mu}{\sqrt\alpha} )^2 ) \rd x\\
        &=\frac{\sqrt\alpha}{2\pi} \int x\bar\Phi(kx)  \exp(-\frac{x^2}{2} - \frac{\alpha}{2}(x - \frac{\delta-\mu}{\sqrt\alpha} )^2 ) \rd x\\
        &=\frac{\sqrt\alpha}{2\pi}\int x\bar\Phi(kx) \exp(-\frac{1+\alpha}{2}x^2 + \sqrt\alpha(\delta-\mu)x - \frac{(\delta-\mu)^2}{2} )\rd x.
    \end{align*}
    Note that the integrand converges pointwise to $x\bar\Phi(kx) \exp(-x^2/2 - (\delta-\mu)^2/2)$ as $\alpha\to0$. It is also bounded. So by deominated convergence theorem, $\EE{c\varphi(c)\bar\Phi(kc)} = O(\sqrt\alpha)$.

    \item The last term is also $O(\sqrt\alpha)$ (same as the third term).
\end{enumerate}

All the terms are of order $O(\sqrt\alpha)$ except the second and the fourth terms, which cancel out.
So the sum of all the terms is $O(\sqrt\alpha)$ as $\alpha\to0$.
Combining this with Equation~\eqref{equ: cond var 2} proves that $\EE{\Var{\cv_\alpha\mid y}}=O(\frac{1}{\sqrt\alpha})$ as $\alpha\to0$.

As $K\to\infty$, the correlation $\rho=-\frac{1}{K-1}\to0$. The only terms above that are unbounded are the fifth and the last terms:
\begin{align*}
    -\frac{1-\rho^2}{\rho} \varphi_\rho(c, c) ,\quad \frac{1}{\rho}\varphi(c)^2.
\end{align*}
Since $\varphi_\rho(c, c) = \varphi(c)^2 + O(\rho)$, the sum of these two terms is $O(1)$. Combined with Equation~\eqref{equ: cond var 2}, we have proved that $\EE{\Var{\cv_\alpha\mid y}}=O(\frac1K)$ as $K\to\infty$.
\end{proof}

To confirm the theoretical result, we conduct a numerical experiment with $\mu=0$ and $\delta=1$. The top panel of Figure~\ref{fig: indicator} shows the squared bias, reducible variance, and MSE versus $1/\alpha$ for $K=8$, for both coupled bootstrap and antithetic CV. For the reducible variance, both the inner variance and outer expectation are estimated by 1000 Monte Carlo samples, respectively. The plot confirms that the reducible variance of antithetic CV is $O(1/\sqrt\alpha)$ while that of coupled bootstrap is $O(1/\alpha)$. The bottom panel of Figure~\ref{fig: indicator} shows the same quantities versus $1/K$ for $\alpha=1/2^8$, confirming that the reducible variance decreases at $O(1/K)$ rate for both methods. The analysis and numerical results show that even for non-weakly differentiable functions, antithetic CV achieves significantly smaller reducible variance than coupled bootstrap.

\begin{figure}
    \begin{subfigure}{\textwidth}
        \includegraphics[width=.9\textwidth]{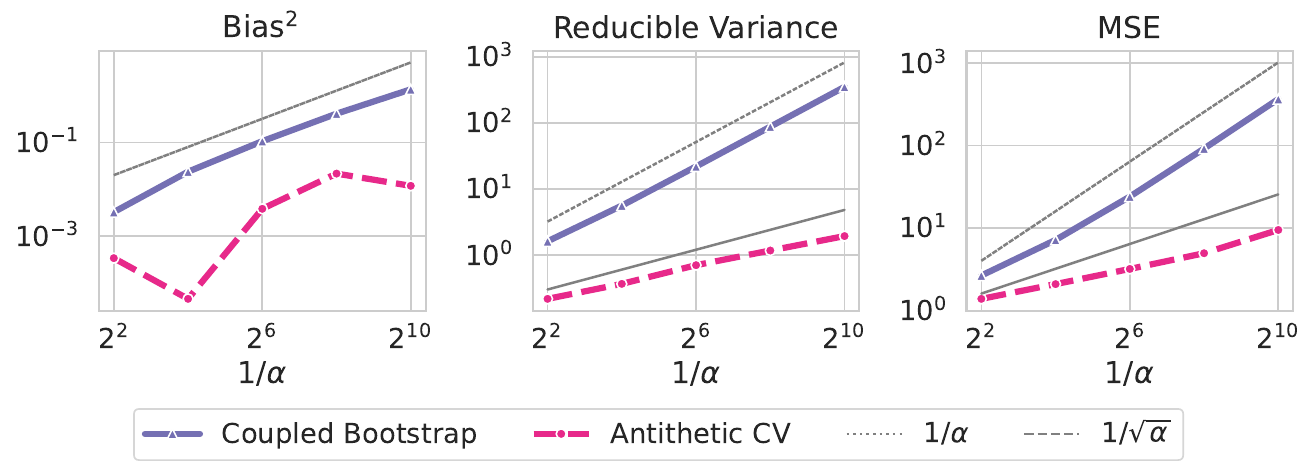}
    \end{subfigure}
    \begin{subfigure}{\textwidth}
        \includegraphics[width=.9\textwidth]{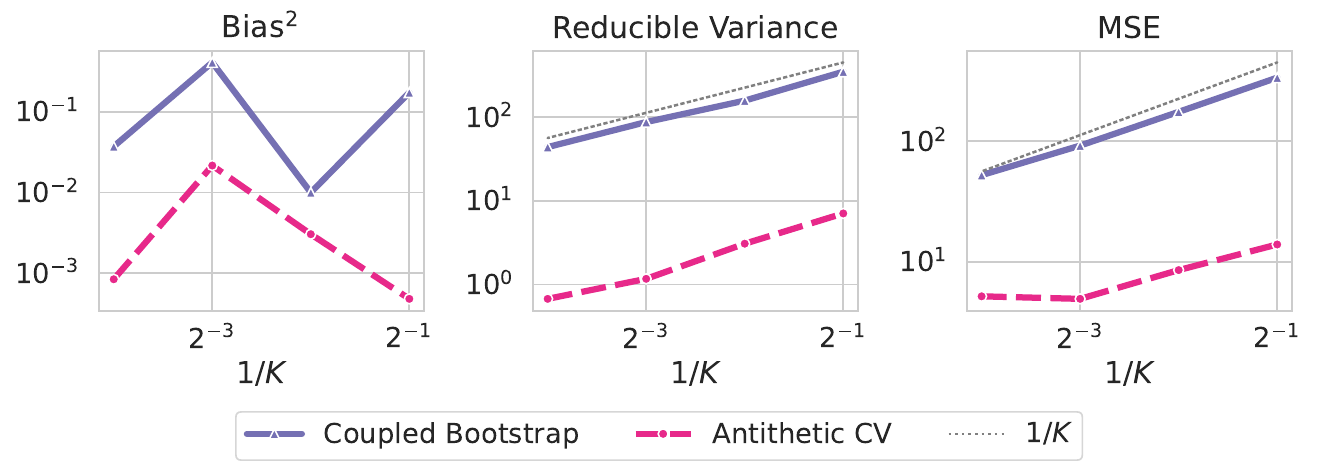}
    \end{subfigure}
    \caption{Squared bias, reducible variance, and MSE of coupled bootstrap and antithetic CV in estimating the prediction error of the indicator function $g(y)=\Indc{y\geq 1}$ with $y\sim\N(0,1)$. The top panel shows the results when $K=8$ and $\alpha$ varies. The dotted and dashed lines represent the $1/\alpha$ and $1/\sqrt\alpha$ rates, respectively. The bottom panel shows the results when $\alpha=1/2^8$ and $K$ varies. The dotted line represents the $1/K$ rate.}
    \label{fig: indicator}
\end{figure}

\section{Additional numerical results}
\label{sec: lasso}
We consider a lasso regression example following the same setup as in \cite{oliveira2021unbiased}. Specifically, we set $n = 100$, and generate $p = 200$ independent features from $\mathcal{N}(0,1)$. The true coefficient vector $\beta$ has five nonzero entries, uniformly sampled from $[-1,1]$. The noise variance $\sigma^2$ is chosen such that $\operatorname{Var}_n(X\beta)/\sigma^2 = 0.4$. The simulation is repeated 200 times.

In addition to standard cross-validation and coupled bootstrap (CB), we compare our estimator with the Breiman-Ye (BY) estimator, as defined in \cite{oliveira2021unbiased}. Notably, this BY estimator includes slight modifications from the original versions introduced in \cite{breiman1992little} and \cite{ye1998measuring}. As discussed in \cite{oliveira2021unbiased}, the BY estimator may exhibit higher irreducible variance than CB when $g$ is unstable. Figure~\ref{fig:lassocv} shows that our estimator achieves smaller error than both the CB and BY estimators. Note that \cite{breiman1992little} and \cite{ye1998measuring} suggest using a larger $\sigma$ (e.g., between 0.6 and 1) and a much larger $K$ (e.g., above 100). In contrast, our estimator attains low error even with $K = 10$, while ensuring small bias when $\alpha$ is small.

\begin{figure}
    \centering
    \includegraphics[width=0.6\linewidth]{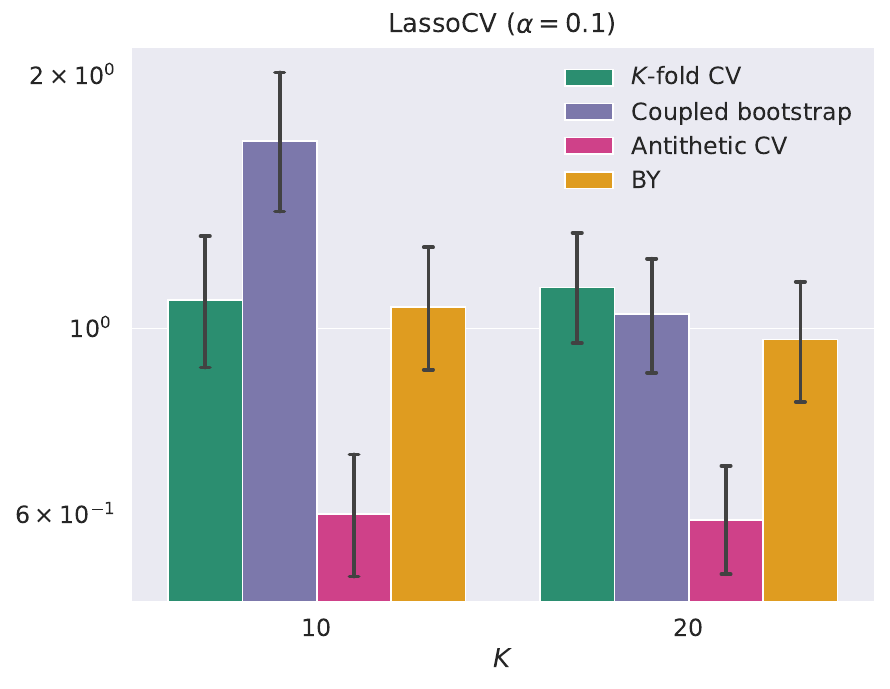}
    \caption{MSE in estimating the prediction error in the lasso regression example. We consider $K=10$ and 20. For CB, antithetic CV, and BY, $\alpha$ is set to be 0.1. }
    \label{fig:lassocv}
\end{figure}

\end{document}